\documentclass[a4paper]{article}
\usepackage{amssymb}
\usepackage{amsfonts}
\usepackage{amsmath}
\usepackage{geometry}
\usepackage{graphicx}

\newtheorem{theorem}{Theorem}
\newtheorem{acknowledgement}[theorem]{Acknowledgement}

\newtheorem{corollary}[theorem]{Corollary}

\newtheorem{lemma}[theorem]{Lemma}

\newtheorem{proposition}[theorem]{Proposition}
\newtheorem{remark}[theorem]{Remark}

\newenvironment{proof}[1][Proof]{\noindent\textbf{#1.} }{\ \rule{0.5em}{0.5em}}
\newenvironment{proof1}[1][Proof]{\noindent\textbf{#1.} }{\ \rule{0.0em}{0.0em}}
\input{tcilatex}
\begin{document}

\title{On a system of equations arising in viscoelasticity theory of
fractional type}
\author{Teodor M. Atanackovic%
\begin{footnote}
{Department of Mechanics, Faculty of Technical Sciences,
University of Novi Sad, Trg D. Obradovica, 6, 21000 Novi Sad,
Serbia, atanackovic@uns.ac.rs}
\end{footnote}, Stevan Pilipovic%
\begin{footnote}
{Department of Mathematics, Faculty of Natural Sciences and
Mathematics, University of Novi Sad, Trg D. Obradovica, 4, 21000
Novi Sad, Serbia, stevan.pilipovic@dmi.uns.ac.rs}
\end{footnote} and Dusan Zorica%
\begin{footnote}
{Mathematical Institute, Serbian Academy of Sciences and Arts,
Kneza Mihaila 36, 11000 Beograd, Serbia, dusan\textunderscore
zorica@mi.sanu.ac.rs}
\end{footnote}}
\maketitle

\begin{abstract}
\noindent We study a system of partial differential equations with integer
and fractional derivatives arising in the study of forced oscillatory motion
of a viscoelastic rod. We propose a new approach considering a quotient of
relations appearing in the constitutive equation instead the constitutive
equation itself. Both, a rod and a body are assumed to have finite mass. The
motion of a body is assumed to be translatory. Existence and uniqueness for
the corresponding initial-boundary value problem is proved within the spaces
of functions and distributions.

\bigskip

\end{abstract}

\section{Introduction}

In this paper we study (\ref{sys-1}) - (\ref{BC}) derived in \cite{APZ-6}.
The system corresponds to a motion of a viscoelastic rod fixed at one end
and of a body of finite mass attached to the other end. Also, an outer
force, having the action line coinciding with the axis of the rod, acts at
the body attached to the free end of a rod. In the dimensionless form the
system of equations, initial and boundary conditions, describing such a
motion, reads%
\begin{gather}
\frac{\partial }{\partial x}\sigma \left( x,t\right) =\kappa ^{2}\frac{%
\partial ^{2}}{\partial t^{2}}u\left( x,t\right) ,\;\;\;\;\varepsilon \left(
x,t\right) =\frac{\partial }{\partial x}u\left( x,t\right) ,\;\;x\in \left[
0,1\right] ,\;t>0,  \label{sys-1} \\
\int_{0}^{1}\phi _{\sigma }\left( \gamma \right) {}_{0}\mathrm{D}%
_{t}^{\gamma }\sigma \left( x,t\right) \mathrm{d}\gamma =\int_{0}^{1}\phi
_{\varepsilon }\left( \gamma \right) {}_{0}\mathrm{D}_{t}^{\gamma
}\varepsilon \left( x,t\right) \mathrm{d}\gamma ,\;\;x\in \left[ 0,1\right]
,\;t>0,  \label{sys-2} \\
u\left( x,0\right) =0,\;\;\;\frac{\partial }{\partial t}u\left( x,0\right)
=0,\;\;\;\sigma \left( x,0\right) =0,\;\;\;\varepsilon \left( x,0\right)
=0,\;\;x\in \left[ 0,1\right] ,  \label{IC} \\
u\left( 0,t\right) =0,\;\;\;\;-\sigma \left( 1,t\right) +F\left( t\right) =%
\frac{\partial ^{2}}{\partial t^{2}}u\left( 1,t\right) ,\;\;t>0.  \label{BC}
\end{gather}

We note that (\ref{sys-1})$_{1}$ represents equation of motion for an
arbitrary material point of a rod. In it, $\sigma $ denotes the stress at
the point $x$ at time $t,$ $u$ is the displacement, $\kappa $ is a constant
representing the ratio between the masses of a rod and a body. In (\ref%
{sys-1})$_{2}$ we use $\varepsilon $ to denote the axial strain of a rod,
while in (\ref{sys-2}) $\phi _{\sigma }$ and $\phi _{\varepsilon }$ denote
constitutive functions or distributions, that are assumed to be known. The
operator of the left Riemann-Liouville fractional derivative of order $%
\gamma \in \left( 0,1\right) $ ${}_{0}\mathrm{D}_{t}^{\gamma }$ is defined
as
\begin{equation*}
_{0}\mathrm{D}_{t}^{\gamma }y\left( t\right) :=\frac{\mathrm{d}}{\mathrm{d}t}%
\left( \frac{t^{-\gamma }}{\Gamma \left( 1-\gamma \right) }\ast y\left(
t\right) \right) ,\;\;t>0,
\end{equation*}%
where $\Gamma $ is the Euler gamma function, $\ast $ is a convolution, i.e.,
if $f,g\in L_{loc}^{1}\left( \mathbf{%
\mathbb{R}
}\right) ,$ $\limfunc{supp}f,g\subset \left[ 0,\infty \right) ,$ then $%
\left( f\ast g\right) \left( t\right) :=\int_{0}^{t}f\left( \tau \right)
g\left( t-\tau \right) \mathrm{d}\tau ,$ $t\in \mathbb{R}$. We refer to \cite%
{TAFDE,Pod,SKM} for a detailed account on fractional calculus. In (\ref{IC})
the initial conditions for an initially undeformed rod are presented.
Finally, (\ref{BC}) represents the boundary conditions corresponding to a
rod with one end fixed at $x=0$ and with force $F$ acting at the body
attached to the other end (at $x=1$). Note that initial-boundary value
problem (\ref{sys-1}) - (\ref{BC}) represents a generalization of a problem
of forced oscillations in the case of a light rod, presented in \cite{a}. We
refer to \cite{APZ-6} for the details for the physical interpretation basis
of system (\ref{sys-1}) - (\ref{BC}).

The main novelty of our approach is that we consider a constitutive equation
in essentially new way. Instead of examining operators acting on stress and
strain in constitutive equation separately, we analyze their quotient after
the application of the Laplace transform to (\ref{sys-2}), that we denote by
$M,$ see (\ref{M}). In this context $M$ appears as a new important quantity
which reflects the inherent properties of a material of a rod. General form
of $M$ requires detailed mathematical analysis which is given in the paper.

The aim of this paper is to prove the existence and uniqueness of a solution
to system (\ref{sys-1}) - (\ref{BC}). Our main results are stated as
Theorems \ref{thmP} and \ref{thmQ} bellow. In proving these theorems we use
several auxiliary results presented in a separate section.

Constitutive equations (\ref{sys-2}) were used earlier in \cite%
{a-2002-a,APZ-3,APZ-4,H-L} in special forms. Also, in the case $\phi
_{\sigma }=\phi _{\varepsilon }$ (\ref{sys-2}) becomes the Hooke Law:%
\begin{equation*}
\sigma \left( x,t\right) =\varepsilon \left( x,t\right) ,\;\;x\in \left[ 0,L%
\right] ,\;t>0.
\end{equation*}%
We note that the constitutive functions or distributions $\phi _{\sigma }$
and $\phi _{\varepsilon }$ appearing in (\ref{sys-2}) must be taken in the
accordance with the Second Law of Thermodynamics. For example, if we take
\begin{equation}
\phi _{\sigma }\left( \gamma \right) :=a^{\gamma },\;\;\phi _{\varepsilon
}\left( \gamma \right) :=b^{\gamma },\;\;\gamma \in \left( 0,1\right)
,\;a,b>0,  \label{a-b}
\end{equation}%
then there is a restriction $a\leq b,$ see \cite{a-2002,a-2003,AKOZ}. The
special case when the constitutive distributions $\phi _{\sigma }$ and $\phi
_{\varepsilon }$ are given by
\begin{equation}
\phi _{\sigma }\left( \gamma \right) :=\delta \left( \gamma \right)
+a\,\delta \left( \gamma -\alpha \right) ,\;\;\phi _{\varepsilon }\left(
\gamma \right) :=\delta \left( \gamma \right) +b\,\delta \left( \gamma
-\alpha \right) ,\;\;\alpha \in \left( 0,1\right) ,\;0<a\leq b,
\label{zener}
\end{equation}%
where $\delta $ is the Dirac distribution, is of particular interest. This
case gives a generalization of the Zener constitutive equation for a
viscoelastic body. The waves in such type of materials were studied in \cite%
{KOZ10}. If $\phi _{\sigma }$ and $\phi _{\varepsilon }$ are given by%
\begin{eqnarray}
&&\phi _{\sigma }\left( \gamma \right)
\begin{tabular}{l}
:=%
\end{tabular}%
\delta \left( \gamma \right) +a\,\delta \left( \gamma -\alpha \right) ,
\notag \\
&&\phi _{\varepsilon }\left( \gamma \right)
\begin{tabular}{l}
:=%
\end{tabular}%
b_{0}\,\delta \left( \gamma -\beta _{0}\right) +b_{1}\,\delta \left( \gamma
-\beta _{1}\right) +b_{2}\,\delta \left( \gamma -\beta _{2}\right) ,
\label{hilf}
\end{eqnarray}%
where $a,$ $b_{0},$ $b_{1},$ $b_{2}$ are positive constants and $0<\alpha
<\beta _{0}<\beta _{1}<\beta _{2}\leq 1,$ then one obtains a constitutive
equation proposed in \cite{hilf}. We note that the system (\ref{sys-1}) - (%
\ref{BC}), with the choice of constitutive functions and distributions (\ref%
{a-b}) and (\ref{zener}), is considered in \cite{APZ-6}. We refer to \cite%
{Mai-10,R-2010,R-S-2010} for the detailed account on the use of fractional
calculus in viscoelasticity.

Note that we can apply our results in the study of behavior of solid-like
materials, as done in \cite{APZ-6} for the cases when $M$ takes the forms (%
\ref{a-b-M}) and (\ref{zener-M}). The behavior of a fluid-like material in a
special form is analyzed in \cite{AKOZ}, where we used constitutive
distributions in the form given by (\ref{hilf}).

The paper is organized as follows.\ We present in \S \thinspace \ref{thms}
main results of our work formulated as Theorems \ref{thmP} and \ref{thmQ}.
Proofs of these theorems are given in two steps: the first one is given at
the beginning of \S \thinspace \ref{prufs-0}, while the second one is given
in \S \thinspace \ref{prufs}. We obtain the displacement $u$\ and stress $%
\sigma $\ as solutions to (\ref{sys-1}) - (\ref{BC}) in the convolution form
by the use of the Laplace transform method. In order to be able to invert
the Laplace transform in \S \thinspace \ref{prufs}, we need several
auxiliary results, that are given in \S \thinspace \ref{AR}. On the basis of
Theorems \ref{thmP} and \ref{thmQ}, we discuss, in \S \thinspace \ref{ER},
Theorems \ref{thmP-el} and \ref{thmQ-el}, a model of elastic rod as a
special case.

\section{Notation and assumptions}

In the sequel we consider analytic functions in%
\begin{equation*}
V=%
\mathbb{C}
\backslash \left( -\infty ,0\right] =\left\{ z=r\mathrm{e}^{\mathrm{i}%
\varphi }\mid r>0,\;\varphi \in \left( -\pi ,\pi \right) \right\} .
\end{equation*}%
We often use notation $\left\vert s\right\vert \rightarrow \infty $ and $%
\left\vert s\right\vert \rightarrow 0$, where we assume that $s\in V.$

The Laplace transform of $f\in L_{loc}^{1}\left( \mathbf{%
\mathbb{R}
}\right) ,$ $f\equiv 0$ in $\left( -\infty ,0\right] $ and $\left\vert
f\left( t\right) \right\vert \leq c\mathrm{e}^{kt},$ $t>0,$ for some $k>0,$
is defined by%
\begin{equation}
\tilde{f}\left( s\right) =\mathcal{L}\left[ f\left( t\right) \right] \left(
s\right) :=\int_{0}^{\infty }f\left( t\right) e^{-st}\mathrm{d}t,\;\;\func{Re%
}s>k  \label{lt}
\end{equation}%
and analytically continued in an appropriate domain (in our case $V$).

We consider spaces of tempered distributions supported by $\bar{%
\mathbb{R}%
}_{+}=\left[ 0,\infty \right) ,$ denoted by $\mathcal{S}_{+}^{\prime }.$ The
Laplace transform of distributions in $\mathcal{S}_{+}^{\prime }$ is derived
from (\ref{lt}), since tempered distributions are derivatives of
polynomially bounded continuous functions. We refer to \cite{vlad} for the
spaces of distributions, as well as for the Laplace and Fourier transforms
in such spaces. We use $C\left( \left[ 0,1\right] ,\mathcal{S}_{+}^{\prime
}\right) $ to denote the space of continuous functions on $\left[ 0,1\right]
$ with values in $\mathcal{S}_{+}^{\prime }.$

In the analysis that follows we shall need the properties of an analytic
function $M,$ defined in appropriate domain $\mathcal{V}\subset
\mathbb{C}
$
\begin{equation}
M\left( s\right) :=\sqrt{\frac{\int_{0}^{1}\phi _{\sigma }\left( \gamma
\right) s^{\gamma }\mathrm{d}\gamma }{\int_{0}^{1}\phi _{\varepsilon }\left(
\gamma \right) s^{\gamma }\mathrm{d}\gamma }},\;\;s\in \mathcal{V}.
\label{M}
\end{equation}
We shall have $\mathcal{V}=V.$ For the cases of constitutive functions (\ref%
{a-b}) and (\ref{zener}), $M$ has the respective forms
\begin{eqnarray}
M\left( s\right) &=&\sqrt{\frac{\ln \left( bs\right) }{\ln \left( as\right) }%
\frac{as-1}{bs-1}},\;\;s\in V,\;a\leq b,  \label{a-b-M} \\
M\left( s\right) &=&\sqrt{\frac{1+as^{\alpha }}{1+bs^{\alpha }}},\;\;s\in
V,\;\alpha \in \left( 0,1\right) ,\;a\leq b.  \label{zener-M}
\end{eqnarray}%
In our analysis the next function has a special role%
\begin{equation}
f\left( s\right) :=sM\left( s\right) \sinh \left( \kappa sM\left( s\right)
\right) +\kappa \cosh \left( \kappa sM\left( s\right) \right) ,\;\;s\in V.
\label{polovi-01}
\end{equation}%
As it will be seen from (\ref{P-tilda}) and (\ref{Q-tilda}), $f$ is a
denominator of functions $\tilde{P}$ and $\tilde{Q},$ which, after the
inversion of the Laplace transform, represent solution kernels of $u$ and $%
\sigma ,$ respectively.

We summarize all the assumptions used throughout the manuscript. Let $M$ be
of the form%
\begin{equation*}
M\left( s\right) =r\left( s\right) +\mathrm{i}h\left( s\right) ,\;\;\text{as}%
\;\;\left\vert s\right\vert \rightarrow \infty .
\end{equation*}
We assume:

\begin{enumerate}
\item[(A1)]
\begin{eqnarray*}
&&\lim_{\left\vert s\right\vert \rightarrow \infty }r\left( s\right)
=c_{\infty }>0,\;\;\lim_{\left\vert s\right\vert \rightarrow \infty }h\left(
s\right) =0,\;\;\lim_{\left\vert s\right\vert \rightarrow 0}M\left( s\right)
=c_{0}, \\
&&\text{for some constants}\;c_{\infty },c_{0}>0.
\end{eqnarray*}
\end{enumerate}

Let $s_{n}=\xi _{n}+\mathrm{i}\zeta _{n},$ $n\in
\mathbb{N}
,$ satisfy the equation%
\begin{equation}
f\left( s\right) =0,\;\;s\in V,  \label{polovi-0}
\end{equation}%
where $f$ is given by (\ref{polovi-01}).

\begin{enumerate}
\item[(A2)] There exists $n_{0}>0,$ such that for $n>n_{0}$%
\begin{eqnarray*}
&&\func{Im}s_{n}\in
\mathbb{R}
_{+}\Rightarrow h\left( s_{n}\right) \leq 0,\;\;\;\;\func{Im}s_{n}\in
\mathbb{R}
_{-}\Rightarrow h\left( s_{n}\right) \geq 0, \\
&&\text{where }h:=\func{Im}M.
\end{eqnarray*}

\item[(A3)] There exist $s_{0}>0$ and $c>0$ such that
\begin{equation*}
\left\vert \frac{\mathrm{d}}{\mathrm{d}s}(sM\left( s\right) )\right\vert
\geq c,\;\;\left\vert s\right\vert >s_{0}.
\end{equation*}

\item[(A4)] For every $\gamma >0$ there exists $\theta >0$ and $s_{0}$ such
that%
\begin{equation*}
\left\vert \left( s+\Delta s\right) M\left( s+\Delta s\right) -sM\left(
s\right) \right\vert \leq \gamma ,\;\;\text{if}\;\;\left\vert \Delta
s\right\vert <\theta \;\;\text{and}\;\;\left\vert s\right\vert >s_{0}.
\end{equation*}
\end{enumerate}

Alternatively to $\left( \mathrm{A2}\right) ,$ in Proposition \ref{pr-left},
we shall consider the following assumption.

\begin{enumerate}
\item[(B)] $\left\vert h\left( s\right) \right\vert \leq \frac{C}{\left\vert
s\right\vert },$ $\left\vert s\right\vert >s_{0},$ for some constants $C>0$
and $s_{0}>0.$
\end{enumerate}

It is shown in \cite{APZ-6} that $\left( \mathrm{A1}\right) $ - $\left(
\mathrm{A4}\right) $ hold for $M$ given by (\ref{a-b-M}) and (\ref{zener-M}).

\section{Theorems on the existence and uniqueness\label{thms}}

Our central results are stated in the next two theorems on the existence,
uniqueness and properties of $u$ and $\sigma .$ Recall, $f$ is given by (\ref%
{polovi-01}) and $s_{n},$ $n\in
\mathbb{N}
,$ are solutions of (\ref{polovi-0}).

\begin{theorem}
\label{thmP}Let $F\in \mathcal{S}_{+}^{\prime }$ and suppose that $M$
satisfies assumptions $\left( \mathrm{A1}\right) $ - $\left( \mathrm{A4}%
\right) .$ Then the unique solution $u$ to (\ref{sys-1}) - (\ref{BC}) is
given by%
\begin{equation}
u\left( x,t\right) =F\left( t\right) \ast P\left( x,t\right) ,\;\;x\in \left[
0,1\right] ,\;t>0,  \label{u}
\end{equation}%
where%
\begin{eqnarray}
P\left( x,t\right) &=&\frac{1}{\pi }\dint\nolimits_{0}^{\infty }\func{Im}%
\left( \frac{M\left( q\mathrm{e}^{-\mathrm{i}\pi }\right) \sinh \left(
\kappa xqM\left( q\mathrm{e}^{-\mathrm{i}\pi }\right) \right) }{qM\left( q%
\mathrm{e}^{-\mathrm{i}\pi }\right) \sinh \left( \kappa qM\left( q\mathrm{e}%
^{-\mathrm{i}\pi }\right) \right) +\kappa \cosh \left( \kappa qM\left( q%
\mathrm{e}^{-\mathrm{i}\pi }\right) \right) }\right) \frac{\mathrm{e}^{-qt}}{%
q}\mathrm{d}q  \notag \\
&&+2\sum_{n=1}^{\infty }\func{Re}\left( \func{Res}\left( \tilde{P}\left(
x,s\right) \mathrm{e}^{st},s_{n}\right) \right) ,\;\;x\in \left[ 0,1\right]
,\;t>0,  \label{P1} \\
P\left( x,t\right) &=&0,\;\;x\in \left[ 0,1\right] ,\;t<0.  \notag
\label{P0}
\end{eqnarray}%
The residues are given by%
\begin{equation}
\func{Res}\left( \tilde{P}\left( x,s\right) \mathrm{e}^{st},s_{n}\right) =%
\left[ \frac{1}{s}\frac{M\left( s\right) \sinh \left( \kappa xsM\left(
s\right) \right) }{\frac{\mathrm{d}}{\mathrm{d}s}f\left( s\right) }\mathrm{e}%
^{st}\right] _{s=s_{n}},\;\;x\in \left[ 0,1\right] ,\;t>0,  \label{res-P}
\end{equation}

Then $P\in C\left( \left[ 0,1\right] \times \left[ 0,\infty \right) \right) $
and $u\in C\left( \left[ 0,1\right] ,\mathcal{S}_{+}^{\prime }\right) .$ In
particular, if $F\in L_{loc}^{1}\left( \left[ 0,\infty \right) \right) ,$
then $u$ is continuous on $\left[ 0,1\right] \times \left[ 0,\infty \right)
. $
\end{theorem}

The following theorem is related to stress $\sigma .$ We formulate this
theorem with $F=H,$ where $H$ denotes the Heaviside function, while the more
general cases of $F$ are discussed in Remark \ref{o-efu}, below.

\begin{theorem}
\label{thmQ}Let $F=H$ and suppose that $M$ satisfies assumptions $\left(
\mathrm{A1}\right) $ - $\left( \mathrm{A4}\right) .$ Then the unique
solution $\sigma _{H}$ to (\ref{sys-1}) - (\ref{BC}), is given by%
\begin{eqnarray}
\sigma _{H}\left( x,t\right) &=&H\left( t\right) +\frac{\kappa }{\pi }%
\dint\nolimits_{0}^{\infty }\func{Im}\left( \frac{\cosh \left( \kappa
xqM\left( q\mathrm{e}^{\mathrm{i}\pi }\right) \right) }{qM\left( q\mathrm{e}%
^{\mathrm{i}\pi }\right) \sinh \left( \kappa qM\left( q\mathrm{e}^{\mathrm{i}%
\pi }\right) \right) +\kappa \cosh \left( \kappa qM\left( q\mathrm{e}^{%
\mathrm{i}\pi }\right) \right) }\right) \frac{\mathrm{e}^{-qt}}{q}\mathrm{d}q
\notag \\
&&+2\sum_{n=1}^{\infty }\func{Re}\left( \func{Res}\left( \tilde{\sigma}%
_{H}\left( x,s\right) \mathrm{e}^{st},s_{n}\right) \right) ,\;\;x\in \left[
0,1\right] ,\;t>0,  \label{Q1} \\
\sigma _{H}\left( x,t\right) &=&0,\;\;x\in \left[ 0,1\right] ,\;t<0.
\label{Q0}
\end{eqnarray}%
The residues are given by%
\begin{equation}
\func{Res}\left( \sigma _{H}\left( x,s\right) \mathrm{e}^{st},s_{n}\right) =%
\left[ \frac{\kappa \cosh \left( \kappa xsM\left( s\right) \right) }{s\frac{%
\mathrm{d}}{\mathrm{d}s}f\left( s\right) }\mathrm{e}^{st}\right]
_{s=s_{n}},\;\;x\in \left[ 0,1\right] ,\;t>0.  \label{res-Q}
\end{equation}%
In particular, $\sigma _{H}$ is continuous on $\left[ 0,1\right] \times %
\left[ 0,\infty \right) .$
\end{theorem}

\begin{remark}
\label{o-efu}\qquad

\begin{enumerate}
\item The assumption $F=H$ in Theorem \ref{thmQ} can be relaxed by requiring
that $F$ is locally integrable and
\begin{equation*}
\tilde{F}\left( s\right) \approx \frac{1}{s^{\alpha }},\;\;\text{as}%
\;\;\left\vert s\right\vert \rightarrow \infty ,
\end{equation*}%
for some $\alpha \in \left( 0,1\right) .$ This condition ensures the
convergence of the series in (\ref{Q1}).

\item If $F=\delta ,$ or even $F\left( t\right) =\frac{\mathrm{d}^{k}}{%
\mathrm{d}t^{k}}\delta \left( t\right) ,$ one uses $\sigma _{H},$ given by (%
\ref{Q1}), in order to obtain $\sigma $ as the $k+1$-th distributional
derivative:%
\begin{equation*}
\sigma =\frac{\mathrm{d}^{k+1}}{\mathrm{d}t^{k+1}}\sigma _{H}\in C\left( %
\left[ 0,1\right] ,\mathcal{S}_{+}^{\prime }\right) .
\end{equation*}
\end{enumerate}
\end{remark}

\section{Proofs of Theorems \protect\ref{thmP} and \protect\ref{thmQ} \label%
{prufs-0}}

Theorems \ref{thmP} and \ref{thmQ} will be proved in two steps.

\begin{proof1}[Step 1]
Applying formally the Laplace transform to (\ref{sys-1}) - (\ref{BC}), we
obtain%
\begin{gather}
\frac{\partial }{\partial x}\tilde{\sigma}\left( x,s\right) =\kappa ^{2}s^{2}%
\tilde{u}\left( x,s\right) ,\;\;\;\;\tilde{\varepsilon}\left( x,s\right) =%
\frac{\partial }{\partial x}\tilde{u}\left( x,s\right) ,\;\;x\in \left[ 0,1%
\right] ,\;s\in V,  \label{S-LT-1} \\
\tilde{\sigma}\left( x,s\right) \int_{0}^{1}\phi _{\sigma }\left( \gamma
\right) s^{\gamma }\mathrm{d}\gamma =\tilde{\varepsilon}\left( x,s\right)
\int_{0}^{1}\phi _{\varepsilon }\left( \gamma \right) s^{\gamma }\mathrm{d}%
\gamma ,\;\;x\in \left[ 0,1\right] ,\;s\in V,  \label{S-LT-2} \\
\tilde{u}\left( 0,s\right) =0,\;\;\tilde{\sigma}\left( 1,s\right) +s^{2}%
\tilde{u}\left( 1,s\right) =\tilde{F}\left( s\right) ,\;\;s\in V.
\label{S-LT-3}
\end{gather}%
By (\ref{S-LT-2}) and (\ref{M}), we have
\begin{equation}
\tilde{\sigma}\left( x,s\right) =\frac{1}{M^{2}\left( s\right) }\tilde{%
\varepsilon}\left( x,s\right) ,\;\;x\in \left[ 0,1\right] ,\;s\in V.
\label{sigma-tilda}
\end{equation}

In order to obtain the displacement $u,$ we use (\ref{S-LT-1}) and (\ref%
{sigma-tilda}) to obtain%
\begin{equation}
\frac{\partial ^{2}}{\partial x^{2}}\tilde{u}\left( x,s\right) -\left(
\kappa sM\left( s\right) \right) ^{2}\tilde{u}\left( x,s\right) =0,\;\;x\in %
\left[ 0,1\right] ,\;s\in V.  \label{zvezda}
\end{equation}%
The solution of (\ref{zvezda}) is%
\begin{equation*}
\tilde{u}\left( x,s\right) =C_{1}\left( s\right) \mathrm{e}^{\kappa
xsM\left( s\right) }+C_{2}\left( s\right) \mathrm{e}^{-\kappa xsM\left(
s\right) },\;\;x\in \left[ 0,1\right] ,\;s\in V,
\end{equation*}%
where $C_{1}$ and $C_{2}$ are arbitrary functions which are determined from (%
\ref{S-LT-3})$_{1}$ as $2C=C_{1}=-C_{2}.$ Therefore,%
\begin{equation}
\tilde{u}\left( x,s\right) =C\left( s\right) \sinh \left( \kappa xsM\left(
s\right) \right) ,\;\;x\in \left[ 0,1\right] ,\;s\in V.  \label{u-tilda-1}
\end{equation}%
By (\ref{M}), (\ref{S-LT-1})$_{2}$, (\ref{sigma-tilda}) and (\ref{u-tilda-1}%
) we have
\begin{equation}
\tilde{\sigma}\left( x,s\right) =C\left( s\right) \frac{\kappa s}{M\left(
s\right) }\cosh \left( \kappa xsM\left( s\right) \right) ,\;\;x\in \left[ 0,1%
\right] ,\;s\in V.  \label{sigma-tilda-1}
\end{equation}%
Using (\ref{u-tilda-1}) and (\ref{sigma-tilda-1}) at $x=1,$ by (\ref{S-LT-3})%
$_{2},$ we obtain%
\begin{equation*}
C\left( s\right) =\frac{M\left( s\right) \tilde{F}\left( s\right) }{s\left(
sM\left( s\right) \sinh \left( \kappa sM\left( s\right) \right) +\kappa
\cosh \left( \kappa sM\left( s\right) \right) \right) },\;\;s\in V.
\end{equation*}%
Therefore, the Laplace transforms of displacement (\ref{u-tilda-1}) and
stress (\ref{sigma-tilda-1}) are
\begin{equation}
\tilde{u}\left( x,s\right) =\tilde{F}\left( s\right) \tilde{P}\left(
x,s\right) \;\;\text{and}\;\;\tilde{\sigma}\left( x,s\right) =\tilde{F}%
\left( s\right) \tilde{Q}\left( x,s\right) ,\;\;x\in \left[ 0,1\right]
,\;s\in V,  \label{u,sigma-tilda}
\end{equation}%
where%
\begin{eqnarray}
\tilde{P}\left( x,s\right) &=&\frac{1}{s}\frac{M\left( s\right) \sinh \left(
\kappa xsM\left( s\right) \right) }{sM\left( s\right) \sinh \left( \kappa
sM\left( s\right) \right) +\kappa \cosh \left( \kappa sM\left( s\right)
\right) },\;\;x\in \left[ 0,1\right] ,\;s\in V,  \label{P-tilda} \\
\tilde{Q}\left( x,s\right) &=&\frac{\kappa \cosh \left( \kappa xsM\left(
s\right) \right) }{sM\left( s\right) \sinh \left( \kappa sM\left( s\right)
\right) +\kappa \cosh \left( \kappa sM\left( s\right) \right) },\;\;x\in %
\left[ 0,1\right] ,\;s\in V.  \label{Q-tilda}
\end{eqnarray}%
Applying the inverse Laplace transform to (\ref{u,sigma-tilda}) we obtain $u$
in the form (\ref{u}), while by using $F=H,$ i.e., $\tilde{F}\left( s\right)
=\frac{1}{s},$ in (\ref{u,sigma-tilda}) we obtain $\sigma _{H}$ in the form (%
\ref{Q1}), (\ref{Q0}).
\end{proof1}

In the sequel we shall justify the formal calculation given above.

\subsection{Auxiliary results\label{AR}}

\subsubsection{Zeros of the function $f$ \label{zeros}}

The following propositions establish the location and the multiplicity of
poles of functions $\tilde{P}$ and $\tilde{Q},$ given by (\ref{P-tilda}) and
(\ref{Q-tilda}), respectively.

\begin{proposition}
\label{pr-conj}Assume $\left( \mathrm{A1}\right) .$ Equation (\ref{polovi-0}%
) has countably many solutions $s_{n},$ $n\in
\mathbb{N}
,$ with the properties
\begin{equation}
s_{n}M\left( s_{n}\right) =\mathrm{i}w_{n},\;\;\tan \left( \kappa
w_{n}\right) =\frac{\kappa }{w_{n}},\;\;w_{n}\in
\mathbb{R}
,\;w_{n}\neq 0.  \label{sn}
\end{equation}%
Complex conjugate $\bar{s}_{n}$ also satisfies (\ref{polovi-0}), $n\in
\mathbb{N}
$.
\end{proposition}

Note $\func{Im}s_{n}\neq 0$ in (\ref{sn}), so all the solutions belong to $%
V. $

\begin{proof}
We seek for the solutions of (\ref{polovi-0}), or equivalently of%
\begin{equation}
\mathrm{e}^{2\kappa sM\left( s\right) }=\frac{sM\left( s\right) -\kappa }{%
sM\left( s\right) +\kappa },\;\;s\in V.  \label{polovi}
\end{equation}%
Put $sM\left( s\right) =v\left( s\right) +\mathrm{i}w\left( s\right) ,$ $%
s\in V,$ where $v,w$ are real-valued functions. Taking the modulus of (\ref%
{polovi}), we obtain%
\begin{equation}
\mathrm{e}^{2\kappa v}=\frac{\left( v-\kappa \right) ^{2}+w^{2}}{\left(
v+\kappa \right) ^{2}+w^{2}}.  \label{polovi-im-1}
\end{equation}%
Fix $w$ and let $v<0.$ Then%
\begin{equation*}
\mathrm{e}^{2\kappa v}<1\;\;\text{and}\;\;\frac{\left( v-\kappa \right)
^{2}+w^{2}}{\left( v+\kappa \right) ^{2}+w^{2}}>1.
\end{equation*}%
Now let $v>0.$ Then%
\begin{equation*}
\mathrm{e}^{2\kappa v}>1\;\;\text{and}\;\;\frac{\left( v-\kappa \right)
^{2}+w^{2}}{\left( v+\kappa \right) ^{2}+w^{2}}<1.
\end{equation*}%
Thus, in both cases we have a contradiction and we conclude that the
solutions to (\ref{polovi-im-1}) satisfy $v(s)=0.$ Therefore, the solutions
of (\ref{polovi}) satisfy%
\begin{equation*}
sM\left( s\right) =\mathrm{i}w\left( s\right) ,\;\;s\in V.
\end{equation*}%
Inserting this into (\ref{polovi-0}) yields%
\begin{equation}
\tan \left( \kappa w\right) =\frac{\kappa }{w},\;\;w\in
\mathbb{R}
.  \label{v,w}
\end{equation}%
Since the tangent function is periodic, we conclude that there are countably
many values of $w,$ denoted by $w_{n},$ $n\in
\mathbb{N}
,$ satisfying (\ref{v,w}). Hence, we have (\ref{sn}).

In order to prove that the solutions $s_{n}\in V$ of (\ref{polovi-0}) are
complex conjugated we note that $M\left( \bar{s}\right) =\overline{M\left(
s\right) },$ $s\in V.$ By (\ref{sn}), $\bar{s}_{n}M\left( \bar{s}_{n}\right)
=\overline{s_{n}M\left( s_{n}\right) }=-\mathrm{i}w_{n}.$ Thus, $\bar{s}_{n}$
also solves (\ref{polovi-0}).
\end{proof}

\begin{proposition}
\label{pr-conj-besk}Assume $\left( \mathrm{A1}\right) .$ Positive (negative)
solutions of $\tan \left( \kappa w_{n}\right) =\frac{\kappa }{w_{n}}$
satisfy $w_{n}\approx \frac{n\pi }{\kappa }$ ($w_{n}\approx -\frac{n\pi }{%
\kappa }$) as $n\rightarrow \infty .$
\end{proposition}

\begin{proof}
As we noted, if $w_{n}$ satisfies (\ref{sn})$,$ then $-w_{n}$ also satisfies
(\ref{sn})$.$ Since $\frac{\kappa }{w_{n}}$ monotonically decreases to zero
for all $w_{n}>0,$ by (\ref{sn}), we have that $\kappa w_{n}$ behave as
zeros of the tangent function, i.e., that
\begin{equation*}
w_{n}\approx \frac{n\pi }{\kappa }\;\;(w_{n}\approx -\frac{n\pi }{\kappa }%
)\;\;\text{as}\;\;n\rightarrow \infty .
\end{equation*}
\end{proof}

\begin{proposition}
\label{pr-left}Assume $\left( \mathrm{A1}\right) ,$ $\left( \mathrm{A2}%
\right) ,$ or $\left( \mathrm{A1}\right) ,$ $\left( \mathrm{B}\right) .$
Then there exist $\xi _{0}>0$ and $n_{0}\in \mathbb{N}$ so that the real
part of $s_{n},$ $n\in
\mathbb{N}
,$ denoted by $\xi _{n}$ satisfies $\xi _{n}<\xi _{0},$ $n>n_{0}.$ Moreover,
if we assume additionally $\left( \mathrm{A3}\right) ,$ then the solutions $%
s_{n}$ of (\ref{polovi-0}) are of multiplicity one for $n>n_{0}.$
\end{proposition}

\begin{proof}
Assume $\left( \mathrm{A1}\right) $ and $\left( \mathrm{A2}\right) .$ By (%
\ref{sn}) we have%
\begin{equation*}
\left( \xi _{n}+\mathrm{i}\zeta _{n}\right) \left( r\left( s_{n}\right) +%
\mathrm{i}h\left( s_{n}\right) \right) \approx \mathrm{i}w_{n},\;\;n>n_{0}.
\end{equation*}%
This implies%
\begin{eqnarray}
&&\xi _{n}r\left( s_{n}\right) =\zeta _{n}h\left( s_{n}\right) ,\;\;n\in
\mathbb{N}
,  \label{r1} \\
&&\xi _{n}h\left( s_{n}\right) +\zeta _{n}r\left( s_{n}\right) \approx
w_{n},\;\;n>n_{0}.  \label{r2}
\end{eqnarray}%
Inserting (\ref{r1}) into (\ref{r2}), we obtain%
\begin{equation}
\zeta _{n}\approx \frac{n\pi }{c_{\infty }\kappa }\;\;\text{because of}%
\;\;w_{n}\approx \frac{n\pi }{\kappa },\;\;n>n_{0},  \label{r3}
\end{equation}%
or
\begin{equation}
\zeta _{n}\approx -\frac{n\pi }{c_{\infty }\kappa }\;\;\text{because of}%
\;\;w_{n}\approx -\frac{n\pi }{\kappa },\;\;n>n_{0}.  \label{r4}
\end{equation}%
In the case of (\ref{r3}), we have%
\begin{equation*}
\xi _{n}\approx \frac{\zeta _{n}}{c_{\infty }}h\left( s_{n}\right) \leq
0,\;\;n>n_{0},
\end{equation*}%
since $s_{n}$ belongs to the upper complex half-plane. In the case of (\ref%
{r4}), we have
\begin{equation*}
\xi _{n}\approx \frac{\zeta _{n}}{c_{\infty }}h\left( s_{n}\right) \leq
0,\;\;n>n_{0},
\end{equation*}%
since then $s_{n}$ belongs to the lower complex half-plane. Thus, in both
cases $\xi _{n}\leq 0$ for sufficiently large $n.$ This proves the first
assertion.

Assume $\left( \mathrm{A1}\right) $ and $\left( \mathrm{B}\right) .$ This
and (\ref{sn}) imply $s_{n}\approx \mathrm{i}\frac{w_{n}}{c_{\infty }},$ for
$n>n_{0}.$ Thus, by (\ref{r1}) and (\ref{r2}), we obtain%
\begin{equation*}
\left\vert \xi _{n}\right\vert \leq \frac{\left\vert w_{n}\right\vert }{%
c_{\infty }^{2}}\frac{C}{\left\vert s_{n}\right\vert }\leq \frac{C}{%
c_{\infty }},\;\;n>n_{0}.
\end{equation*}%
So, the real parts $\xi _{n}$ of solutions $s_{n}$ of (\ref{polovi-0})
satisfy $\xi _{n}\in \left[ -\frac{C}{c_{\infty }},\frac{C}{c_{\infty }}%
\right] ,$ for $n>n_{0}.$ This is even a stronger condition for the zeros,
but it will not be used in the sequel.

In order to prove that the solutions $s_{n},$ $n>n_{0},$ of $f$ are of
multiplicity one, we use $\left( \mathrm{A3}\right) ,$ differentiate (\ref%
{polovi-0}) and obtain
\begin{equation*}
\frac{\mathrm{d}f\left( s\right) }{\mathrm{d}s}=\left( \left( 1+\kappa
^{2}\right) \sinh \left( \kappa sM\left( s\right) \right) +\kappa sM\left(
s\right) \cosh \left( \kappa sM\left( s\right) \right) \right) \frac{\mathrm{%
d}}{\mathrm{d}s}\left( sM\left( s\right) \right) ,\;\;s\in V.
\end{equation*}%
Calculating the previous expression at $s_{n},$ we have that%
\begin{equation*}
\left\vert \left. \frac{\mathrm{d}f\left( s\right) }{\mathrm{d}s}\right\vert
_{s=s_{n}}\right\vert =\left\vert \left( 1+\kappa ^{2}\right) \sin \left(
\kappa w_{n}\right) +\kappa w_{n}\cos \left( \kappa w_{n}\right) \right\vert
\left\vert \left[ \frac{\mathrm{d}}{\mathrm{d}s}\left( sM\left( s\right)
\right) \right] _{s=s_{n}}\right\vert
\end{equation*}%
is different from zero by $\left( \mathrm{A3}\right) .$
\end{proof}

\subsubsection{Estimates\label{estimates}}

In the sequel, we assume $\left( \mathrm{A1}\right) $ - $\left( \mathrm{A4}%
\right) .$

Let $R>0.$ A quarter of a disc, denoted by $D,$ and its boundary $\Gamma $
are defined by%
\begin{eqnarray*}
D &=&D_{R}=\left\{ s=\rho \,\mathrm{e}^{\mathrm{i}\varphi }\mid \rho \leq
R,\;\varphi \in \left( \frac{\pi }{2},\pi \right) \right\} , \\
\Gamma &=&\Gamma _{R}=\left\{ s=R\,\mathrm{e}^{\mathrm{i}\varphi }\mid
\varphi \in \left( \frac{\pi }{2},\pi \right) \right\} .
\end{eqnarray*}%
Let $S$ be the set of all solutions of (\ref{polovi-0}) in $D.$

In the calculation of $P$ and $Q$ in the next subsections, we shall need the
estimates given in the next two lemmas. Recall, $S$ is the set of zeros.

\begin{lemma}
\label{lema0}Let $\eta >0$ and
\begin{equation*}
D_{\eta }=\left\{ s\in D\mid \left\vert s-s_{j}\right\vert >\eta ,\;s_{j}\in
S\right\} .
\end{equation*}%
Then there exist $s_0>0$ and $p_{\eta }>0,$ such that
\begin{equation}
\left\vert f\left( s\right) \right\vert >p_{\eta },\;\;\text{if}\;\;s\in
D_{\eta },\; \left\vert s\right\vert>s_0.  \label{jedan}
\end{equation}
\end{lemma}

\begin{remark}
We shall have in the sequel that certain assertions hold for $n>n_{0}.$ This
is related to the subindexes of the solutions to (\ref{polovi-0}), but it
also implies that we consider domains in $D$ where $\left\vert s\right\vert
>s_{0},$ where $s_{0}$ depends on $n_{0}.$
\end{remark}

\begin{proof}
If (\ref{jedan}) does not hold, then there exists a sequence $\left\{ \tilde{%
s}_{n}\right\} _{n\in
\mathbb{N}
}\in D_{\eta }$ such that
\begin{equation}
\left\vert f\left( \tilde{s}_{n}\right) \right\vert =\eta _{n}\rightarrow
0,\;\;n\rightarrow \infty .  \label{pet}
\end{equation}%
This implies%
\begin{equation*}
\left\vert \func{Re}\left( \tilde{s}_{n}M\left( \tilde{s}_{n}\right) \right)
\right\vert \rightarrow 0,\;\;\;\;\left\vert \func{Im}\left( \tilde{s}%
_{n}M\left( \tilde{s}_{n}\right) \right) \right\vert \rightarrow \infty ,\;\;%
\text{as}\;\;n\rightarrow \infty .
\end{equation*}%
Our aim is to show that there exist $N\in
\mathbb{N}
$ and $s_{N}\in S,$ so that
\begin{equation}
\left\vert \tilde{s}_{N}-s_{N}\right\vert \leq \eta  \label{tri}
\end{equation}%
and this will be the contradiction. Let $\delta <\frac{\pi }{c_{\infty
}\kappa }$ and $\delta \ll \eta $. Recall that there exists $n_{0},$ such
that for $s_{n}\in S,$ $n>n_{0}$, there holds%
\begin{equation*}
\func{Re}\left( s_{n}M\left( s_{n}\right) \right) =0,\;\;\left\vert \func{Im}%
\left( s_{n}M\left( s_{n}\right) \right) -\frac{n\pi }{c_{\infty }\kappa }%
\right\vert <\frac{\delta }{2}.
\end{equation*}%
Now consider the intervals
\begin{equation*}
I_{n}=\left( \frac{n\pi }{c_{\infty }\kappa }+\frac{\delta }{2},\frac{\left(
n+1\right) \pi }{c_{\infty }\kappa }-\frac{\delta }{2}\right)
,\;\;I_{n+1}=\left( \frac{\left( n+1\right) \pi }{c_{\infty }\kappa }+\frac{%
\delta }{2},\frac{\left( n+2\right) \pi }{c_{\infty }\kappa }-\frac{\delta }{%
2}\right) ,\ldots .
\end{equation*}%
Since $\left\vert \func{Re}\left( \tilde{s}_{n}M\left( \tilde{s}_{n}\right)
\right) \right\vert \rightarrow 0,$ we see that
\begin{equation*}
\left\vert \func{Im}\left( \tilde{s}_{n}M\left( \tilde{s}_{n}\right) \right)
\right\vert \in I_{n}\cup I_{n+1}\cup \ldots ,\;\;n>n_{0}.
\end{equation*}%
Put $\kappa \tilde{s}_{n}M\left( \tilde{s}_{n}\right) =t_{n}+\mathrm{i}\tau
_{n},$ $n\in
\mathbb{N}
.$ We have
\begin{eqnarray}
2\kappa f\left( \tilde{s}_{n}\right) &=&t_{n}\left( \mathrm{e}^{t_{n}}-%
\mathrm{e}^{-t_{n}}\right) \cos \tau _{n}-\tau _{n}\left( \mathrm{e}^{t_{n}}+%
\mathrm{e}^{-t_{n}}\right) \sin \tau _{n}+\kappa ^{2}\left( \mathrm{e}%
^{t_{n}}+\mathrm{e}^{-t_{n}}\right) \cos \tau _{n}  \notag \\
&&+\mathrm{i}\left( t_{n}\left( \mathrm{e}^{t_{n}}+\mathrm{e}%
^{-t_{n}}\right) \sin \tau _{n}+\tau _{n}\left( \mathrm{e}^{t_{n}}-\mathrm{e}%
^{-t_{n}}\right) \cos \tau _{n}+\kappa ^{2}\left( \mathrm{e}^{t_{n}}-\mathrm{%
e}^{-t_{n}}\right) \sin \tau _{n}\right) .  \notag \\
&&  \label{cetiri}
\end{eqnarray}%
Put $\tau _{n}=k_{n}\pi +r_{n},$ $r_{n}<\pi ,$ where $k_{n},n\in \mathbb{N}$
is an increasing sequence of natural numbers. We shall show that $r_{n}$
must have a subsequence tending to zero and this will lead to (\ref{tri}),
i.e., to the contradiction with $\tilde{s}_{n}\in D_{\eta }$.

So let us assume that $r_{n}$ does not have a subsequence converging to
zero; so $r_{n}>\frac{\delta }{2},$ $n>n_{0}.$ Having in mind that $%
t_{n}\rightarrow 0,$ and dropping summands tending to zero in (\ref{cetiri}%
), one obtains ($n\rightarrow \infty $)
\begin{eqnarray*}
|2\kappa f\left( \tilde{s}_{n}\right) | &\sim &|-\tau _{n}\left( \mathrm{e}%
^{t_{n}}+\mathrm{e}^{-t_{n}}\right) \sin \tau _{n}+\kappa ^{2}\left( \mathrm{%
e}^{t_{n}}+\mathrm{e}^{-t_{n}}\right) \cos \tau _{n}| \\
&=&|-\left( k_{n}\pi +r_{n}\right) \left( \mathrm{e}^{t_{n}}+\mathrm{e}%
^{-t_{n}}\right) \sin r_{n}+\kappa ^{2}\left( \mathrm{e}^{t_{n}}+\mathrm{e}%
^{-t_{n}}\right) \cos r_{n}| \\
&\sim &|-k_{n}\pi \left( \mathrm{e}^{t_{n}}+\mathrm{e}^{-t_{n}}\right) \sin
r_{n}-r_{n}\left( \mathrm{e}^{t_{n}}+\mathrm{e}^{-t_{n}}\right) \sin
r_{n}+\kappa ^{2}\left( \mathrm{e}^{t_{n}}+\mathrm{e}^{-t_{n}}\right) \cos
r_{n}|.
\end{eqnarray*}%
Now we see that the first summand on the last right hand side tends to
infinity, while the second and the third one are bounded. This is in
contradiction with (\ref{pet}), so lemma is proved.
\end{proof}

In the next proposition we shall find estimates on $f$ needed for the later
calculation of integrals.

\begin{proposition}
\label{estimreal}\qquad

\begin{enumerate}
\item[$(i)$] Let $D_0$ be a subdomain of $D$. If $\left\vert \func{Re}\left(
sM(s)\right) \right\vert >d,$ $s\in D_{0}\subset D$, then there exist $c>0$
and $s_{0}>0$ such that
\begin{equation*}
\left\vert f(s)\right\vert \geq c|sM(s)|\left\vert \sinh \left( \kappa
sM\left( s\right) \right) \right\vert ,\;\;s\in D_{0},\;\left\vert
s\right\vert >s_{0}.
\end{equation*}

\item[$(ii)$] If $s\in D_{\eta },$ $\left\vert s\right\vert >s_{0}$ (see
Lemma \ref{lema0}), then $\left\vert f(s)\right\vert \geq p_{\eta }.$
\end{enumerate}
\end{proposition}

Note that in the case $\left( ii\right) $ condition $\left\vert \func{Re}%
\left( sM(s)\right) \right\vert \leq d$ is not assumed, although we consider
this case in $\left( ii\right) ,$ since this part is already a consequence
of Lemma \ref{lema0}.

\begin{proof}
$\left( i\right) $ follows from the fact that $sM(s)\sinh (\kappa sM(s))$
tends faster to the infinity than $\cosh (\kappa sM(s))$ when $\left\vert
s\right\vert \rightarrow \infty ,$ $s\in D_{0}.$ So, for some $c>0$ and $%
\left\vert s\right\vert >s_{0},$
\begin{eqnarray*}
\left\vert f\left( s\right) \right\vert &\geq &\left\vert sM\left( s\right)
\right\vert \left\vert \sinh (\kappa sM(s))\right\vert -\kappa
^{2}\left\vert \cosh \left( \kappa sM\left( s\right) \right) \right\vert \\
&\geq &c\left\vert sM\left( s\right) \right\vert \left\vert \sinh (\kappa
sM(s))\right\vert .
\end{eqnarray*}%
This implies the assertion.
\end{proof}

We need one more estimate of $f$ in the case when we have to control how $s$
is close to the zero set $S$ of $f$. This is needed for the small
deformation of the circle arc $\Gamma _{R}$ near the point of $\Gamma _{R}$
which is close to some zero of $f$. We need assumptions $\left( \mathrm{A3}%
\right) $ and $\left( \mathrm{A4}\right) .$

For the later use we choose $\varepsilon >0$ such that $\varepsilon <\frac{%
\theta }{2}$ ($\theta $ is from $\left( \mathrm{A4}\right) $) and that the
difference $\left\vert s_{1}M(s_{1})-s_{2}M(s_{2})\right\vert \leq \gamma $
implies small differences
\begin{equation*}
\left\vert \cosh (s_{1}M(s_{1}))-\cosh (s_{2}M(s_{2}))\right\vert <\delta
_{1},\;\left\vert \sinh (s_{1}M(s_{1}))-\sinh (s_{2}M(s_{2}))\right\vert
<\delta _{1},
\end{equation*}%
as we shall need in the proof of the next lemma (in (\ref{nejedn})).

\begin{lemma}
\label{lema}Let $0<\varepsilon <\frac{\theta }{2}.$ Then there exist $%
n_{0}\in
\mathbb{N}
$ and $d>0$ such that for $s_j\in S$
\begin{equation}
j>n_{0},\;\varepsilon <\left\vert s-s_{j}\right\vert \leq 2\varepsilon
\Rightarrow \func{Re}\left( sM\left( s\right) \right) >d.  \label{de}
\end{equation}
\end{lemma}

\begin{proof}
Since for suitable $n_{0}$%
\begin{equation*}
\left\vert s_{n+1}-s_{n}\right\vert \geq \left\vert \func{Im}\left(
s_{n+1}-s_{n}\right) \right\vert \approx \frac{\pi }{c_{\infty }\kappa }>%
\frac{\pi }{2c_{\infty }\kappa },\;\;n>n_{0},
\end{equation*}%
we have that the balls $L\left( s_{j},2\varepsilon \right) $ are disjoint
for $j>n_{0}.$ Let $j>n_{0}$ and $\varepsilon <\left\vert s-s_{j}\right\vert
\leq 2\varepsilon .$ By the Taylor formula we have%
\begin{equation}
\left\vert f\left( s\right) -f\left( s_{j}\right) \right\vert =\left\vert
\frac{\mathrm{d}f\left( \bar{s}\right) }{\mathrm{d}s}\right\vert \left\vert
s-s_{j}\right\vert >\varepsilon \left\vert \frac{\mathrm{d}f\left( \bar{s}%
-s_{j}\right) }{\mathrm{d}s}\right\vert ,\;\;\varepsilon <\left\vert \bar{s}%
\right\vert \leq 2\varepsilon .  \label{tejlor}
\end{equation}%
So with $n_{0}$ large enough we have $\left\vert s\right\vert >s_{0}$ so
that $\left( \mathrm{A3}\right) $ implies%
\begin{equation*}
\left\vert \frac{\mathrm{d}}{\mathrm{d}s}\left( sM\left( s\right) \right)
\right\vert \geq c,\;\;\text{for}\;\;\left\vert s\right\vert >s_{0}.
\end{equation*}%
In the sequel we shall refer to the following set of conditions
\begin{equation}
j>n_{0},\;|s|>s_{0},\;\varepsilon <\left\vert s-s_{j}\right\vert \leq
2\varepsilon ,\;\;\varepsilon <\left\vert \bar{s}\right\vert \leq
2\varepsilon .  \label{as}
\end{equation}%
Assuming (\ref{as}), it follows%
\begin{eqnarray}
\left\vert \frac{\mathrm{d}f\left( \bar{s}\right) }{\mathrm{d}s}\right\vert
&=&\left\vert \left( 1+\kappa ^{2}\right) \sinh \left( \kappa \bar{s}M\left(
\bar{s}\right) \right) +\kappa \bar{s}M\left( \bar{s}\right) \cosh \left(
\kappa \bar{s}M\left( \bar{s}\right) \right) \right\vert \left\vert \left[
\frac{\mathrm{d}}{\mathrm{d}s}\left( sM\left( s\right) \right) \right] _{s=%
\bar{s}}\right\vert  \notag \\
&\geq &c\left( \kappa \left\vert \bar{s}M\left( \bar{s}\right) \right\vert
\left\vert \cosh \left( \kappa \bar{s}M\left( \bar{s}\right) \right)
\right\vert -\left( 1+\kappa ^{2}\right) \left\vert \sinh \left( \kappa \bar{%
s}M\left( \bar{s}\right) \right) \right\vert \right) .  \label{izvod-f}
\end{eqnarray}%
Now, we estimate $\left\vert f\right\vert $, assuming (\ref{as}), and obtain%
\begin{equation}
\left\vert f\left( s\right) \right\vert \leq \left\vert sM\left( s\right)
\right\vert \left\vert \sinh \left( \kappa sM\left( s\right) \right)
\right\vert +\kappa \left\vert \cosh \left( \kappa sM\left( s\right) \right)
\right\vert .  \label{f}
\end{equation}%
Using (\ref{izvod-f}), (\ref{f}) with (\ref{as}) in (\ref{tejlor}), we have%
\begin{eqnarray}
&&\left\vert sM\left( s\right) \right\vert \left\vert \sinh \left( \kappa
sM\left( s\right) \right) \right\vert +\kappa \left\vert \cosh \left( \kappa
sM\left( s\right) \right) \right\vert  \notag \\
&&\;\;\;\;\;\;\;\;>\varepsilon c\left( \kappa \left\vert \bar{s}M\left( \bar{%
s}\right) \right\vert \left\vert \cosh \left( \kappa \bar{s}M\left( \bar{s}%
\right) \right) \right\vert -\left( 1+\kappa ^{2}\right) \left\vert \sinh
\left( \kappa \bar{s}M\left( \bar{s}\right) \right) \right\vert \right) .
\label{z}
\end{eqnarray}%
The final part of the proof of the lemma is to show that (\ref{z}) implies
that there exists $d$ such that (\ref{de}) holds if (\ref{as}) is satisfied.
Contrary to (\ref{de}), assume that there exist sequences $\tilde{s}_{n},$ $%
s_{j_{n}}$ and $d_{n}\rightarrow 0$ such that%
\begin{equation*}
\left\vert \func{Re}\left( \tilde{s}_{n}\right) M\left( \tilde{s}_{n}\right)
\right\vert \leq d_{n},\;\;\text{if}\;\;\varepsilon <\left\vert \tilde{s}%
_{n}-s_{j_{n}}\right\vert \leq 2\varepsilon ,n>n_{0}.
\end{equation*}%
Since $\tilde{s}_{n},$ $n>n_{0}$ satisfies (\ref{as}), by (\ref{z}), we
would have%
\begin{eqnarray}
&&\left\vert \tilde{s}_{n}M\left( \tilde{s}_{n}\right) \right\vert
\left\vert \sinh \left( \kappa \tilde{s}_{n}M\left( \tilde{s}_{n}\right)
\right) \right\vert +\kappa \left\vert \cosh \left( \kappa \tilde{s}%
_{n}M\left( \tilde{s}_{n}\right) \right) \right\vert  \notag \\
&&\;\;\;\;\;\;\;\;>\varepsilon c\left( \kappa \left\vert \bar{s}_{n}M\left(
\bar{s}_{n}\right) \right\vert \left\vert \cosh \left( \kappa \bar{s}%
_{n}M\left( \bar{s}_{n}\right) \right) \right\vert -\left( 1+\kappa
^{2}\right) \left\vert \sinh \left( \kappa \bar{s}_{n}M\left( \bar{s}%
_{n}\right) \right) \right\vert \right) .  \label{nejedn}
\end{eqnarray}%
The second addend on the right hand side tends to zero, thus we neglect it.
Moreover, we note that $\kappa \left\vert \cosh \left( \kappa \tilde{s}%
_{n}M\left( \tilde{s}_{n}\right) \right) \right\vert $ cannot majorize $%
\varepsilon c\kappa \left\vert \bar{s}_{n}M\left( \bar{s}_{n}\right)
\right\vert \left\vert \cosh \left( \kappa \bar{s}_{n}M\left( \bar{s}%
_{n}\right) \right) \right\vert $ since the second one tends to infinity
while the first one is finite. Thus, in (\ref{nejedn}), leading terms on
both sides are the first ones (with $\tilde{s}_{n}M\left( \tilde{s}%
_{n}\right) $ and $\bar{s}_{n}M\left( \bar{s}_{n}\right) $) and we skip the
second terms on both sides of (\ref{nejedn}). It follows, with another $%
c_{0}>0,$
\begin{equation}
\left\vert \tilde{s}_{n}M\left( \tilde{s}_{n}\right) \right\vert \left\vert
\sinh \left( \kappa \tilde{s}_{n}M\left( \tilde{s}_{n}\right) \right)
\right\vert \geq c_{0}\left\vert \bar{s}_{n}M\left( \bar{s}_{n}\right)
\right\vert \left\vert \cosh \left( \kappa \bar{s}_{n}M\left( \bar{s}%
_{n}\right) \right) \right\vert .  \label{nov}
\end{equation}%
Now according to (\ref{nov}) we can choose $\gamma $ (which then determines $%
\theta $) and $\varepsilon $ so that, with suitable $\delta _{1}$ and $c_{1}$%
, (\ref{nejedn}) implies
\begin{equation*}
\left\vert \sinh \left( \kappa \tilde{s}_{n}M\left( \tilde{s}_{n}\right)
\right) \right\vert >c_{1}\left\vert \cosh \left( \kappa \bar{s}_{n}M\left(
\bar{s}_{n}\right) \right) \right\vert ,
\end{equation*}%
on the domain (\ref{as}). However, this leads to the contradiction, since $%
\left\vert \sinh \left( \kappa \tilde{s}_{n}M\left( \tilde{s}_{n}\right)
\right) \right\vert \rightarrow 0,$ while $\left\vert \cosh \left( \kappa
\bar{s}_{n}M\left( \bar{s}_{n}\right) \right) \right\vert $ is close to one,
under the assumptions. This proves the lemma.
\end{proof}

We note that we can choose $\eta $ in Lemma \ref{lema0} to be equal to $%
2\varepsilon $ of Lemma \ref{lema}. In this way we obtain that Lemma \ref%
{lema} and Proposition \ref{estimreal} imply:

\begin{proposition}
\label{zak}\qquad

\begin{enumerate}
\item[$\left( i\right) $] There exists $\varepsilon >0$ and $d>0$ such that,
for $\varepsilon <\left\vert s-s_{n}\right\vert \leq 2\varepsilon ,$ $%
s_{n}\in S$ and $\left\vert s\right\vert >s_{0}$
\begin{equation*}
\left\vert \func{Re}\left( sM\left( s\right) \right) \right\vert >d
\end{equation*}%
and
\begin{equation*}
\left\vert f(s)\right\vert \geq c\left\vert sM(s)\right\vert \left\vert
\sinh (\kappa sM(s))\right\vert ,\;\;\varepsilon <\left\vert
s-s_{n}\right\vert \leq 2\varepsilon ,\;s_{n}\in S,\;\left\vert s\right\vert
>s_{0}.
\end{equation*}

\item[$\left( ii\right) $] Let $\left\vert s-s_{n}\right\vert >2\varepsilon
, $ $s_{n}\in S,$ $\left\vert s\right\vert >s_{0}$ and $d,\varepsilon ,c$ be
as in $\left( i\right) $.

\begin{enumerate}
\item[$a)$] If $\left\vert \func{Re}\left( sM\left( s\right) \right)
\right\vert >d$, then
\begin{equation*}
\left\vert f(s)\right\vert \geq c\left\vert sM(s)\right\vert \left\vert
\sinh (\kappa sM(s))\right\vert .
\end{equation*}

\item[$b)$] If $\left\vert \func{Re}\left( sM\left( s\right) \right)
\right\vert \leq d$, then
\begin{equation*}
\left\vert f(s)\right\vert >p_{2\varepsilon }
\end{equation*}%
(see the comment before the proof of Proposition \ref{estimreal}).
\end{enumerate}
\end{enumerate}
\end{proposition}

With the notation of the previous proposition, we finally come to
corollaries which will be used in the subsequent subsections. We keep the
notation from Proposition \ref{zak}.

\begin{corollary}
\label{zbognjegasveovo}\qquad

\begin{enumerate}
\item[$\left( i\right) $] There exists $\varepsilon >0$ and $C>0$ such that
for $\varepsilon <\left\vert s-s_{n}\right\vert \leq 2\varepsilon ,$ $%
s_{n}\in S$ and $\left\vert s\right\vert >s_{0},$
\begin{equation*}
\frac{\left\vert sM(s)\sinh \left( \kappa xsM\left( s\right) \right)
\right\vert }{\left\vert f(s)\right\vert }\leq \frac{1}{c}\mathrm{e}%
^{-\kappa \left( 1-x\right) \func{Re}\left( sM\left( s\right) \right) }\leq
C;
\end{equation*}

\item[$\left( ii\right) $] for $\left\vert s-s_{n}\right\vert >2\varepsilon
, $ $s_{n}\in S,$ $\left\vert s\right\vert >s_{0}$
\begin{equation*}
\frac{\left\vert sM(s)\sinh \left( \kappa xsM\left( s\right) \right)
\right\vert }{\left\vert f(s)\right\vert }\leq \frac{1}{c}\mathrm{e}%
^{-\kappa \left( 1-x\right) \func{Re}\left( sM\left( s\right) \right) }\leq
C,
\end{equation*}%
or
\begin{equation*}
\frac{\left\vert sM(s)\sinh \left( \kappa xsM\left( s\right) \right)
\right\vert }{\left\vert f(s)\right\vert }\leq \frac{d\mathrm{e}^{\kappa
xd}\left\vert 1-\mathrm{e}^{-2\kappa xsM\left( s\right) }\right\vert }{%
p_{2\varepsilon }}\leq C.
\end{equation*}
\end{enumerate}
\end{corollary}

Since we also need to estimate $\frac{\left\vert \cosh \left( \kappa
xsM\left( s\right) \right) \right\vert }{\left\vert f(s)\right\vert },$ we
use again Proposition \ref{zak}. Moreover, we use the fact that in the case $%
\left\vert \func{Re}\left( sM\left( s\right) \right) \right\vert >d$, there
exists $c>0$ such that
\begin{equation*}
\left\vert \cosh \left( \kappa xsM\left( s\right) \right) \right\vert \leq
c\left\vert \sinh \left( \kappa xsM\left( s\right) \right) \right\vert
,\;\;\left\vert s\right\vert \rightarrow \infty .
\end{equation*}

\begin{corollary}
\label{IIzbognjegasveovo}\qquad

\begin{enumerate}
\item[$\left( i\right) $] There exists $\varepsilon >0$ and $C>0$ such that
for $\varepsilon <\left\vert s-s_{n}\right\vert \leq 2\varepsilon ,$ $%
s_{n}\in S$ $\left\vert s\right\vert >s_{0},$
\begin{equation*}
\frac{\left\vert \cosh \left( \kappa xsM\left( s\right) \right) \right\vert
}{\left\vert f(s)\right\vert }\leq \frac{1}{c\cdot c_{\infty }}\frac{1}{%
\left\vert s\right\vert }\mathrm{e}^{-\kappa \left( 1-x\right) \func{Re}%
\left( sM\left( s\right) \right) }\leq \frac{C}{\left\vert s\right\vert };
\end{equation*}

\item[$\left( ii\right) $] for $\left\vert s-s_{n}\right\vert >2\varepsilon
, $ $s_{n}\in S,$ $\left\vert s\right\vert >s_{0}$
\begin{equation*}
\frac{\left\vert \cosh \left( \kappa xsM\left( s\right) \right) \right\vert
}{\left\vert f(s)\right\vert }\leq \frac{1}{c\cdot c_{\infty }}\frac{1}{%
\left\vert s\right\vert }\mathrm{e}^{-\kappa \left( 1-x\right) \func{Re}%
\left( sM\left( s\right) \right) }\leq \frac{C}{\left\vert s\right\vert },
\end{equation*}%
or
\begin{equation*}
\frac{\left\vert \cosh \left( \kappa xsM\left( s\right) \right) \right\vert
}{\left\vert f(s)\right\vert }\leq \frac{\mathrm{e}^{\kappa xd}\left\vert 1+%
\mathrm{e}^{-2\kappa xsM\left( s\right) }\right\vert }{p_{2\varepsilon }}%
\leq C.
\end{equation*}
\end{enumerate}
\end{corollary}

\subsection{Continuation of the proofs of Theorems \protect\ref{thmP} and
\protect\ref{thmQ} \label{prufs}}

In this section we finish the proofs of Theorems \ref{thmP} and \ref{thmQ}.
First, we finish the proof of Theorem \ref{thmP}.

\begin{proof}[Proof of Theorem \protect\ref{thmP}. Step 2]
We calculate $P\left( x,t\right) ,$ $x\in \left[ 0,1\right] ,$ $t\in
\mathbb{R}
,$ by the integration over the contour given in Figure \ref{fig}.
\begin{figure}[h]
\centering
\includegraphics[scale=0.45]{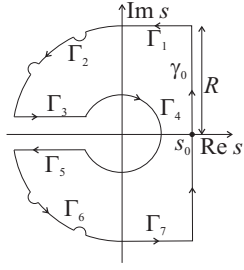}
\caption{Integration contour $\Gamma $}
\label{fig}
\end{figure}
Small, inside or outside half-circles, depending on the zeros of $f$ near $%
\Gamma _{2}$ and $\Gamma _{6},$ have radius $\varepsilon $ determined in
Corollaries \ref{zbognjegasveovo} and \ref{IIzbognjegasveovo}. This will be
explained in the proof. The Cauchy residues theorem yields
\begin{equation}
\oint\nolimits_{\Gamma }\tilde{P}\left( x,s\right) \mathrm{e}^{st}\mathrm{d}%
s=2\pi \mathrm{i}\sum_{n=1}^{\infty }\left( \func{Res}\left( \tilde{P}\left(
x,s\right) \mathrm{e}^{st},s_{n}\right) +\func{Res}\left( \tilde{P}\left(
x,s\right) \mathrm{e}^{st},\bar{s}_{n}\right) \right) ,  \label{KF-P}
\end{equation}%
where $\Gamma =\Gamma _{1}\cup \Gamma _{2}\cup \Gamma _{3}\cup \Gamma
_{4}\cup \Gamma _{5}\cup \Gamma _{6}\cup \Gamma _{7}\cup \gamma _{0},$ so
that poles of $\tilde{P}$ lie inside the contour $\Gamma .$

First we show that the series of residues in (\ref{P1}) is real-valued and
convergent. Proposition \ref{pr-left} implies that the poles $s_{n},$ $n\in
\mathbb{N}
,$ of $\tilde{P},$ given by (\ref{P-tilda}), are simple for sufficiently
large $n$. Then the residues in (\ref{KF-P}) can be calculated using (\ref%
{res-P}) as%
\begin{eqnarray}
\func{Res}\left( \tilde{P}\left( x,s\right) \mathrm{e}^{st},s_{n}\right) &=&%
\left[ \frac{M\left( s\right) \sinh \left( \kappa xsM\left( s\right) \right)
}{\left( 1+\kappa ^{2}\right) \sinh \left( \kappa sM\left( s\right) \right)
+\kappa sM\left( s\right) \cosh \left( \kappa sM\left( s\right) \right) }%
\right] _{s=s_{n}}  \notag \\
&&\times \left[ \frac{\mathrm{e}^{st}}{s\frac{\mathrm{d}}{\mathrm{d}s}\left(
sM\left( s\right) \right) }\right] _{s=s_{n}},\;\;x\in \left[ 0,1\right]
,\;t>0.  \label{rez}
\end{eqnarray}%
Substituting (\ref{sn}) in (\ref{rez}), one obtains ($x\in \left[ 0,1\right]
,$ $t>0$)%
\begin{equation*}
\func{Res}\left( \tilde{P}\left( x,s\right) \mathrm{e}^{st},s_{n}\right) =%
\frac{w_{n}\sin \left( \kappa w_{n}x\right) }{\left( 1+\kappa ^{2}\right)
\sin \left( \kappa w_{n}\right) +\kappa w_{n}\cos \left( \kappa w_{n}\right)
}\frac{\mathrm{e}^{s_{n}t}}{\left[ s^{2}\frac{\mathrm{d}}{\mathrm{d}s}\left(
sM\left( s\right) \right) \right] _{s=s_{n}}}.
\end{equation*}%
Proposition \ref{pr-left} implies%
\begin{equation*}
\left\vert \mathrm{e}^{s_{n}t}\right\vert <\bar{c}\mathrm{e}^{at},\;\;t>0,
\end{equation*}%
for some $a\in
\mathbb{R}
.$ Since $\left\vert s_{n}\right\vert \approx \frac{n\pi }{c_{\infty }\kappa
}$ and $\left\vert w_{n}\right\vert \approx \frac{n\pi }{\kappa },$ for $%
n>n_{0},\ $it follows%
\begin{equation*}
\left\vert \func{Res}\left( \tilde{P}\left( x,s\right) \mathrm{e}%
^{st},s_{n}\right) \right\vert \leq \frac{\bar{c}}{\kappa }\frac{\mathrm{e}%
^{at}}{\left\vert \left[ s^{2}\frac{\mathrm{d}}{\mathrm{d}s}\left( sM\left(
s\right) \right) \right] _{s=s_{n}}\right\vert },\;\;x\in \left[ 0,1\right]
,\;t>0.
\end{equation*}%
Now we use assumption $\left( \mathrm{A3}\right) .$ This implies that there
exists $K>0$ such that%
\begin{equation*}
\left\vert \func{Res}\left( \tilde{P}\left( x,s\right) \mathrm{e}%
^{st},s_{n}\right) \right\vert \leq K\frac{\mathrm{e}^{at}}{n^{2}},\;\;x\in %
\left[ 0,1\right] ,\;t>0.
\end{equation*}%
This implies that the series of residues in (\ref{KF-P}), i.e., in (\ref{P1}%
) is convergent.

Now, we calculate the integral over $\Gamma $ in (\ref{KF-P}). First, we
consider the integral along contour $\Gamma _{1}=\left\{ s=p+\mathrm{i}R\mid
p\in \left[ 0,s_{0}\right] ,\;R>0\right\} ,$ where $R$ is defined as
follows. Take $n_{0}$ so that $\left\vert \func{Im}s_{n}-\frac{n\pi }{%
c_{\infty }\kappa }\right\vert <\eta ,$ where $0<\eta \ll \frac{1}{2}\frac{%
\pi }{c_{\infty }\kappa },$ for $n>n_{0},$ and put%
\begin{equation}
R=\frac{n\pi }{c_{\infty }\kappa }+\frac{1}{2}\frac{\pi }{c_{\infty }\kappa }%
,\;\;n>n_{0}.  \label{R}
\end{equation}%
By (\ref{P-tilda}) and Corollary \ref{zbognjegasveovo}, we have
\begin{equation}
\left\vert \tilde{P}\left( x,s\right) \right\vert \leq \frac{C}{\left\vert
s\right\vert ^{2}},\;\;\left\vert s\right\vert \rightarrow \infty .
\label{P-s-besk}
\end{equation}%
Using (\ref{P-s-besk}), we calculate the integral over $\Gamma _{1}$ as%
\begin{eqnarray*}
\lim_{R\rightarrow \infty }\left\vert \int\nolimits_{\Gamma _{1}}\tilde{P}%
\left( x,s\right) \mathrm{e}^{st}\mathrm{d}s\right\vert &\leq
&\lim_{R\rightarrow \infty }\int_{0}^{s_{0}}\left\vert \tilde{P}\left( x,p+%
\mathrm{i}R\right) \right\vert \left\vert \mathrm{e}^{\left( p+\mathrm{i}%
R\right) t}\right\vert \mathrm{d}p \\
&\leq &C\lim_{R\rightarrow \infty }\int_{0}^{s_{0}}\frac{1}{R^{2}}\mathrm{e}%
^{pt}\mathrm{d}p=0,\;\;x\in \left[ 0,1\right] ,\;t>0.
\end{eqnarray*}%
Similar arguments are valid for the integral along contour $\Gamma _{7}.$
Thus, we have%
\begin{equation*}
\lim\limits_{R\rightarrow \infty }\left\vert \int\nolimits_{\Gamma _{7}}%
\tilde{P}\left( x,s\right) \mathrm{e}^{st}\mathrm{d}s\right\vert =0,\;\;x\in %
\left[ 0,1\right] ,\;t>0.
\end{equation*}

Next, we consider the integral along contour $\Gamma _{2}.$ As it is said
for Figure \ref{fig}, the contour $\Gamma _{2}$ consists of parts of contour
$\Gamma _{R}=\left\{ s=R\,\mathrm{e}^{\mathrm{i}\phi }\mid \phi \in \left[
\frac{\pi }{2},\pi \right] \right\} $ and of finite number of contours $%
\Gamma _{\varepsilon }=\left\{ \left\vert s-s_{k}\right\vert =\varepsilon
\mid f\left( s_{k}\right) =0\right\} $ encircling the poles $s_{k},$ either
from inside, or from outside of $\Gamma _{R}.$ (Note that the distances
between poles are greater than $\varepsilon ,$ $n>n_{0}$.) More precisely,
if a pole is inside of $D_{R},$ then $\Gamma _{\varepsilon }$ is outside of $%
D_{R}$ and if a pole is outside of $D_{R},$ then $\Gamma _{\varepsilon }$ is
inside of $D_{R}.$ By (\ref{P-s-besk}), the integral over the contour $%
\Gamma _{2}$ becomes%
\begin{eqnarray*}
\lim_{R\rightarrow \infty }\left\vert \int\nolimits_{\Gamma _{2}}\tilde{P}%
\left( x,s\right) \mathrm{e}^{st}\mathrm{d}s\right\vert &\leq
&\lim_{R\rightarrow \infty }\int\nolimits_{\frac{\pi }{2}}^{\pi }\left\vert
\tilde{P}\left( x,R\mathrm{e}^{\mathrm{i}\phi }\right) \right\vert
\left\vert \mathrm{e}^{Rt\mathrm{e}^{\mathrm{i}\phi }}\right\vert \left\vert
\mathrm{i}R\mathrm{e}^{\mathrm{i}\phi }\right\vert \mathrm{d}\phi \\
&\leq &C\lim_{R\rightarrow \infty }\int\nolimits_{\frac{\pi }{2}}^{\pi }%
\frac{1}{R}\mathrm{e}^{Rt\cos \phi }\mathrm{d}\phi =0,\;\;x\in \left[ 0,1%
\right] ,\;t>0,
\end{eqnarray*}%
since $\cos \phi \leq 0$ for $\phi \in \left[ \frac{\pi }{2},\pi \right] .$
Similar arguments are valid for the integral along $\Gamma _{6}.$ Thus, we
have%
\begin{equation*}
\lim\limits_{R\rightarrow \infty }\left\vert \int\nolimits_{\Gamma _{6}}%
\tilde{P}\left( x,s\right) \mathrm{e}^{st}\mathrm{d}s\right\vert =0,\;\;x\in %
\left[ 0,1\right] ,\;t>0.
\end{equation*}

Consider the integral along $\Gamma _{4}.$ Let $\left\vert s\right\vert
\rightarrow 0.$ Then, by $\left( \mathrm{A1}\right) ,$ $sM\left( s\right)
\rightarrow 0,$ $\cosh \left( \kappa sM\left( s\right) \right) \rightarrow
1, $ $\sinh \left( \kappa sM\left( s\right) \right) \rightarrow 0$ and $%
\sinh \left( \kappa xsM\left( s\right) \right) \approx \kappa xsM\left(
s\right) .$ Hence, from (\ref{P-tilda}) we have
\begin{equation}
\left\vert \tilde{P}\left( x,s\right) \right\vert \approx x\left\vert
M\left( s\right) \right\vert ^{2}\approx c_{0}^{2}x,\;\;x\in \left[ 0,1%
\right] ,\;s\in V,|s|\rightarrow 0.  \label{P za s nula}
\end{equation}%
The integration along contour $\Gamma _{4}$ gives%
\begin{eqnarray*}
\lim_{r\rightarrow 0}\left\vert \int\nolimits_{\Gamma _{4}}\tilde{P}\left(
x,s\right) \mathrm{e}^{st}\mathrm{d}s\right\vert &\leq &\lim_{r\rightarrow
0}\int\nolimits_{-\pi }^{\pi }\left\vert \tilde{P}\left( x,r\mathrm{e}^{%
\mathrm{i}\phi }\right) \right\vert \left\vert \mathrm{e}^{rt\mathrm{e}^{%
\mathrm{i}\phi }}\right\vert \left\vert \mathrm{i}r\mathrm{e}^{\mathrm{i}%
\phi }\right\vert \mathrm{d}\phi \\
&\leq &c_{0}^{2}x\lim_{r\rightarrow 0}\int\nolimits_{-\pi }^{\pi }r\mathrm{e}%
^{rt\cos \phi }\mathrm{d}\phi =0,\;\;x\in \left[ 0,1\right] ,\;t>0,
\end{eqnarray*}%
where we used (\ref{P za s nula}).

Integrals along $\Gamma _{3},$ $\Gamma _{5}$ and $\gamma _{0}$ give ($x\in %
\left[ 0,1\right] ,$ $t>0$)%
\begin{eqnarray}
\lim_{\substack{ R\rightarrow \infty  \\ r\rightarrow 0}}\int\nolimits_{%
\Gamma _{3}}\tilde{P}\left( x,s\right) \mathrm{e}^{st}\mathrm{d}s
&=&\int\nolimits_{0}^{\infty }\frac{M\left( q\mathrm{e}^{\mathrm{i}\pi
}\right) \sinh \left( \kappa xqM\left( q\mathrm{e}^{\mathrm{i}\pi }\right)
\right) \mathrm{e}^{-qt}}{q\left( qM\left( q\mathrm{e}^{\mathrm{i}\pi
}\right) \sinh \left( \kappa qM\left( q\mathrm{e}^{\mathrm{i}\pi }\right)
\right) +\kappa \cosh \left( \kappa qM\left( q\mathrm{e}^{\mathrm{i}\pi
}\right) \right) \right) }\mathrm{d}q,  \notag \\
&&  \label{P-plus} \\
\lim_{\substack{ R\rightarrow \infty  \\ r\rightarrow 0}}\int\nolimits_{%
\Gamma _{5}}\tilde{P}\left( x,s\right) \mathrm{e}^{st}\mathrm{d}s
&=&-\int\nolimits_{0}^{\infty }\frac{M\left( q\mathrm{e}^{-\mathrm{i}\pi
}\right) \sinh \left( \kappa xqM\left( q\mathrm{e}^{-\mathrm{i}\pi }\right)
\right) \mathrm{e}^{-qt}}{q\left( qM\left( q\mathrm{e}^{-\mathrm{i}\pi
}\right) \sinh \left( \kappa qM\left( q\mathrm{e}^{-\mathrm{i}\pi }\right)
\right) +\kappa \cosh \left( \kappa qM\left( q\mathrm{e}^{-\mathrm{i}\pi
}\right) \right) \right) }\mathrm{d}q,  \notag \\
&&  \label{P-minus} \\
\lim_{R\rightarrow \infty }\int\nolimits_{\gamma _{0}}\tilde{P}\left(
x,s\right) \mathrm{e}^{st}\mathrm{d}s &=&2\pi \mathrm{i}P\left( x,t\right) .
\label{P-pi}
\end{eqnarray}%
We note that (\ref{P-pi}) is valid if the inversion of the Laplace transform
exists, which is true since all the singularities of $\tilde{P}$ are left
from the line $\gamma _{0}$ and the estimates on $\tilde{P}$ over $\gamma
_{0}$ imply the convergence of the integral. Summing up (\ref{P-plus}), (\ref%
{P-minus}) and (\ref{P-pi}) we obtain the left hand side of (\ref{KF-P}) and
finally $P$ in the form given by (\ref{P1}). Analyzing separately%
\begin{eqnarray*}
&&\frac{1}{\pi }\dint\nolimits_{0}^{\infty }\func{Im}\left( \frac{M\left( q%
\mathrm{e}^{-\mathrm{i}\pi }\right) \sinh \left( \kappa xqM\left( q\mathrm{e}%
^{-\mathrm{i}\pi }\right) \right) }{qM\left( q\mathrm{e}^{-\mathrm{i}\pi
}\right) \sinh \left( \kappa qM\left( q\mathrm{e}^{-\mathrm{i}\pi }\right)
\right) +\kappa \cosh \left( \kappa qM\left( q\mathrm{e}^{-\mathrm{i}\pi
}\right) \right) }\right) \frac{\mathrm{e}^{-qt}}{q}\mathrm{d}q, \\
&&2\sum_{n=1}^{\infty }\func{Re}\left( \func{Res}\left( \tilde{P}\left(
x,s\right) \mathrm{e}^{st},s_{n}\right) \right) ,
\end{eqnarray*}%
we conclude that both terms appearing in (\ref{P1}) are continuous functions
on $t\in \left[ 0,\infty \right) ,$ for every $x\in \left[ 0,1\right] .$ The
continuity also holds with respect to $x\in \left[ 0,1\right] $ if we fix $%
t\in \left[ 0,\infty \right) $. This implies that $u$ is a continuous
function on $\left[ 0,1\right] \times \left[ 0,\infty \right) .$ From the
uniqueness of the Laplace transform it follows that $u$ is unique. Since $F$
belongs to $\mathcal{S}_{+}^{\prime }$, it follows that%
\begin{equation*}
u\left( x,\cdot \right) =F\left( \cdot \right) \ast P\left( x,\cdot \right)
\in \mathcal{S}_{+}^{\prime },
\end{equation*}%
for every $x\in \left[ 0,1\right] $ and $u\in C\left( \left[ 0,1\right] ,%
\mathcal{S}_{+}^{\prime }\right) .$ Moreover, if $F\in L_{loc}^{1}\left( %
\left[ 0,\infty \right) \right) ,$ then $u\in C\left( \left[ 0,1\right]
\times \left[ 0,\infty \right) \right) ,$ since $P$ is continuous.
\end{proof}

Next, we complete the proof of Theorem \ref{thmQ}.

\begin{proof}[Proof of Theorem \protect\ref{thmQ}. Step 2]
Since $\tilde{F}\left( s\right) =\tilde{H}\left( s\right) =\frac{1}{s},$ $%
s\neq 0,$ by (\ref{u,sigma-tilda}) and (\ref{Q-tilda}) we obtain%
\begin{equation}
\tilde{\sigma}_{H}\left( x,s\right) =\frac{1}{s}\frac{\kappa \cosh \left(
\kappa xsM\left( s\right) \right) }{sM\left( s\right) \sinh \left( \kappa
sM\left( s\right) \right) +\kappa \cosh \left( \kappa sM\left( s\right)
\right) },\;\;x\in \left[ 0,1\right] ,\;s\in V.  \label{sigma-h-tilda}
\end{equation}%
We calculate $\sigma _{H}\left( x,t\right) ,$ $x\in \left[ 0,1\right] ,$ $%
t\in
\mathbb{R}
,$ by the integration over the same contour from Figure \ref{fig}. The
Cauchy residues theorem yields
\begin{equation}
\oint\nolimits_{\Gamma }\tilde{\sigma}_{H}\left( x,s\right) \mathrm{e}^{st}%
\mathrm{d}s=2\pi \mathrm{i}\sum_{n=1}^{\infty }\left( \func{Res}\left(
\tilde{\sigma}_{H}\left( x,s\right) \mathrm{e}^{st},s_{n}\right) +\func{Res}%
\left( \tilde{\sigma}_{H}\left( x,s\right) \mathrm{e}^{st},\bar{s}%
_{n}\right) \right) ,  \label{KF-Q}
\end{equation}%
so that poles of $\tilde{Q}$ lie inside the contour $\Gamma $.

First we show that the series of residues in (\ref{Q1}) is convergent and
real-valued. The poles $s_{n},$ $n\in
\mathbb{N}
,$ of $\tilde{\sigma}_{H},$ given by (\ref{sigma-h-tilda}) are the same as
for the function $\tilde{P},$ (\ref{P-tilda}). By Proposition \ref{pr-left}
we have that the poles $s_{n},$ $n\in
\mathbb{N}
,$ are simple for sufficiently large $n.$ Then, for $n>n_{0},$ the residues
in (\ref{KF-Q}) can be calculated using (\ref{res-Q}) as%
\begin{eqnarray*}
\func{Res}\left( \tilde{\sigma}_{H}\left( x,s\right) \mathrm{e}%
^{st},s_{n}\right) &=&\left[ \frac{\kappa \cosh \left( \kappa xsM\left(
s\right) \right) }{\left( 1+\kappa ^{2}\right) \sinh \left( \kappa sM\left(
s\right) \right) +\kappa sM\left( s\right) \cosh \left( \kappa sM\left(
s\right) \right) }\right] _{s=s_{n}} \\
&&\times \left[ \frac{\mathrm{e}^{st}}{s\frac{\mathrm{d}}{\mathrm{d}s}\left(
sM\left( s\right) \right) }\right] _{s=s_{n}},\;\;x\in \left[ 0,1\right]
,\;t>0.
\end{eqnarray*}%
By the use of (\ref{sn}) we obtain ($x\in \left[ 0,1\right] ,$ $t>0$)%
\begin{equation*}
\func{Res}\left( \tilde{\sigma}_{H}\left( x,s\right) \mathrm{e}%
^{st},s_{n}\right) =\frac{\kappa \cos \left( \kappa w_{n}x\right) }{\left(
1+\kappa ^{2}\right) \sin \left( \kappa w_{n}\right) +\kappa w_{n}\cos
\left( \kappa w_{n}\right) }\frac{\mathrm{e}^{s_{n}t}}{\left[ s\frac{\mathrm{%
d}}{\mathrm{d}s}\left( sM\left( s\right) \right) \right] _{s=s_{n}}}.
\end{equation*}%
Proposition \ref{pr-left} implies%
\begin{equation*}
\left\vert \mathrm{e}^{s_{n}t}\right\vert <\bar{c}\mathrm{e}^{at},\;\;t>0,
\end{equation*}%
for some $a\in
\mathbb{R}
.$ Since $\left\vert s_{n}\right\vert \approx \frac{n\pi }{c_{\infty }\kappa
}$ and $\left\vert w_{n}\right\vert \approx \frac{n\pi }{\kappa },$ for $%
n>n_{0},\ $it follows%
\begin{equation*}
\left\vert \func{Res}\left( \tilde{\sigma}_{H}\left( x,s\right) \mathrm{e}%
^{st},s_{n}\right) \right\vert \leq \frac{\kappa }{n\pi }\frac{\bar{c}%
\mathrm{e}^{at}}{\left\vert \left[ s\frac{\mathrm{d}}{\mathrm{d}s}\left(
sM\left( s\right) \right) \right] _{s=s_{n}}\right\vert },\;\;x\in \left[ 0,1%
\right] ,\;t>0,\;n>n_{0}.
\end{equation*}%
Now we use assumption $\left( \mathrm{A3}\right) $ and conclude that there
exists $K>0$ such that
\begin{equation*}
\left\vert \func{Res}\left( \tilde{\sigma}_{H}\left( x,s\right) \mathrm{e}%
^{st},s_{n}\right) \right\vert \leq K\frac{\mathrm{e}^{at}}{n^{2}},\;\;x\in %
\left[ 0,1\right] ,\;t>0,\;n>n_{0}.
\end{equation*}%
This implies that the series of residues in (\ref{KF-Q}) (i.e., in (\ref{Q1}%
)) is convergent.

Let us calculate the integral over $\Gamma $ in (\ref{KF-Q}). Consider the
integral along contour%
\begin{equation*}
\Gamma _{1}=\left\{ s=p+\mathrm{i}R\mid p\in \left[ 0,s_{0}\right]
,\;R>0\right\} ,
\end{equation*}%
where $R$ is defined by (\ref{R}). We use estimates%
\begin{equation}
\frac{\left\vert \cosh \left( \kappa xsM\left( s\right) \right) \right\vert
}{\left\vert f(s)\right\vert }\leq \frac{C}{\left\vert s\right\vert },\;\;%
\text{or}\;\;\frac{\left\vert \cosh \left( \kappa xsM\left( s\right) \right)
\right\vert }{\left\vert f(s)\right\vert }\leq C,  \label{estim}
\end{equation}%
from Corollary \ref{IIzbognjegasveovo} in (\ref{sigma-h-tilda}). With the
first estimate in (\ref{estim}) we calculate the integral over $\Gamma _{1}$
as
\begin{eqnarray*}
\lim_{R\rightarrow \infty }\left\vert \int\nolimits_{\Gamma _{1}}\tilde{%
\sigma}_{H}\left( x,s\right) \mathrm{e}^{st}\mathrm{d}s\right\vert &\leq
&\lim_{R\rightarrow \infty }\int_{0}^{s_{0}}\left\vert \tilde{\sigma}%
_{H}\left( x,p+\mathrm{i}R\right) \right\vert \left\vert \mathrm{e}^{\left(
p+\mathrm{i}R\right) t}\right\vert \mathrm{d}p \\
&\leq &C\kappa \lim_{R\rightarrow \infty }\int_{0}^{s_{0}}\frac{1}{R^{2}}%
\mathrm{e}^{pt}\mathrm{d}p=0,\;\;x\in \left[ 0,1\right] ,\;t>0,
\end{eqnarray*}%
while with the second estimate in (\ref{estim}), we have%
\begin{eqnarray*}
\lim_{R\rightarrow \infty }\left\vert \int\nolimits_{\Gamma _{1}}\tilde{%
\sigma}_{H}\left( x,s\right) \mathrm{e}^{st}\mathrm{d}s\right\vert &\leq
&\lim_{R\rightarrow \infty }\int_{0}^{s_{0}}\left\vert \tilde{\sigma}%
_{H}\left( x,p+\mathrm{i}R\right) \right\vert \left\vert \mathrm{e}^{\left(
p+\mathrm{i}R\right) t}\right\vert \mathrm{d}p \\
&\leq &C\kappa \lim_{R\rightarrow \infty }\int_{0}^{s_{0}}\frac{1}{R}\mathrm{%
e}^{pt}\mathrm{d}p=0,\;\;x\in \left[ 0,1\right] ,\;t>0,
\end{eqnarray*}%
Similar arguments are valid for the integral along $\Gamma _{7}.$ Thus, we
have%
\begin{equation*}
\lim\limits_{R\rightarrow \infty }\left\vert \int\nolimits_{\Gamma _{7}}%
\tilde{\sigma}_{H}\left( x,s\right) \mathrm{e}^{st}\mathrm{d}s\right\vert
=0,\;\;x\in \left[ 0,1\right] ,\;t>0.
\end{equation*}

Next, we consider the integral along contour $\Gamma _{2},$ defined as in
the proof of Theorem \ref{thmP}. With the first estimate in (\ref{estim}) we
have that the integral over $\Gamma _{2}$ becomes
\begin{eqnarray*}
\lim_{R\rightarrow \infty }\left\vert \int\nolimits_{\Gamma _{2}}\tilde{%
\sigma}_{H}\left( x,s\right) \mathrm{e}^{st}\mathrm{d}s\right\vert &\leq
&\lim_{R\rightarrow \infty }\int\nolimits_{\frac{\pi }{2}}^{\pi }\left\vert
\tilde{\sigma}_{H}\left( x,R\mathrm{e}^{\mathrm{i}\phi }\right) \right\vert
\left\vert \mathrm{e}^{Rt\mathrm{e}^{\mathrm{i}\phi }}\right\vert \left\vert
\mathrm{i}R\mathrm{e}^{\mathrm{i}\phi }\right\vert \mathrm{d}\phi \\
&\leq &C\kappa \lim_{R\rightarrow \infty }\int\nolimits_{\frac{\pi }{2}%
}^{\pi }\frac{1}{R}\mathrm{e}^{Rt\cos \phi }\mathrm{d}\phi =0,\;\;x\in \left[
0,1\right] ,\;t>0,
\end{eqnarray*}%
since $\cos \phi \leq 0$ for $\phi \in \left[ \frac{\pi }{2},\pi \right] .$
The integral over $\Gamma _{2},$ in the case of the second estimate in (\ref%
{estim}) becomes%
\begin{eqnarray*}
\lim_{R\rightarrow \infty }\left\vert \int\nolimits_{\Gamma _{2}}\tilde{%
\sigma}_{H}\left( x,s\right) \mathrm{e}^{st}\mathrm{d}s\right\vert &\leq
&\lim_{R\rightarrow \infty }\int\nolimits_{\frac{\pi }{2}}^{\pi }\left\vert
\tilde{\sigma}_{H}\left( x,R\mathrm{e}^{\mathrm{i}\phi }\right) \right\vert
\left\vert \mathrm{e}^{Rt\mathrm{e}^{\mathrm{i}\phi }}\right\vert \left\vert
\mathrm{i}R\mathrm{e}^{\mathrm{i}\phi }\right\vert \mathrm{d}\phi \\
&\leq &C\kappa \lim_{R\rightarrow \infty }\int\nolimits_{\frac{\pi }{2}%
}^{\pi }\mathrm{e}^{Rt\cos \phi }\mathrm{d}\phi =0,\;\;x\in \left[ 0,1\right]
,\;t>0,
\end{eqnarray*}%
since $\cos \phi \leq 0$ for $\phi \in \left[ \frac{\pi }{2},\pi \right] .$
Similar arguments are valid for the integral along $\Gamma _{6}.$ Thus, we
have%
\begin{equation*}
\lim\limits_{R\rightarrow \infty }\left\vert \int\nolimits_{\Gamma _{6}}%
\tilde{\sigma}_{H}\left( x,s\right) \mathrm{e}^{st}\mathrm{d}s\right\vert
=0,\;\;x\in \left[ 0,1\right] ,\;t>0.
\end{equation*}

Consider the integral along contour $\Gamma _{4}.$ Let $\left\vert
s\right\vert \rightarrow 0.$ Then, by $\left( \mathrm{A1}\right) ,$ $%
sM\left( s\right) \rightarrow 0,$ $\cosh \left( \kappa sM\left( s\right)
\right) \rightarrow 1,$ $\sinh \left( \kappa sM\left( s\right) \right)
\rightarrow 0$ and $\cosh \left( \kappa xsM\left( s\right) \right)
\rightarrow 1.$ Hence, from (\ref{sigma-h-tilda}) we have%
\begin{equation*}
s\sigma _{H}\left( x,s\right) \approx 1.
\end{equation*}%
The integration along contour $\Gamma _{4}$ gives
\begin{eqnarray}
\lim_{r\rightarrow 0}\int\nolimits_{\Gamma _{4}}\tilde{\sigma}_{H}\left(
x,s\right) \mathrm{e}^{st}\mathrm{d}s &=&\lim_{r\rightarrow
0}\int\nolimits_{\pi }^{-\pi }\tilde{\sigma}_{H}\left( x,r\mathrm{e}^{%
\mathrm{i}\phi }\right) \mathrm{e}^{rt\mathrm{e}^{\mathrm{i}\phi }}\mathrm{i}%
r\mathrm{e}^{\mathrm{i}\phi }\mathrm{d}\phi  \notag \\
&=&\mathrm{i}\int\nolimits_{\pi }^{-\pi }\mathrm{d}\phi =-2\pi \mathrm{i}%
,\;\;x\in \left[ 0,1\right] ,\;t>0.  \label{Q-epsilon}
\end{eqnarray}

Integrals along $\Gamma _{3},$ $\Gamma _{5}$ and $\gamma _{0}$ give ($x\in %
\left[ 0,1\right] ,$ $t>0$)%
\begin{eqnarray}
\lim_{\substack{ R\rightarrow \infty  \\ r\rightarrow 0}}\int\nolimits_{%
\Gamma _{3}}\tilde{\sigma}_{H}\left( x,s\right) \mathrm{e}^{st}\mathrm{d}s
&=&-\int\nolimits_{0}^{\infty }\frac{\kappa \cosh \left( \kappa xqM\left( q%
\mathrm{e}^{\mathrm{i}\pi }\right) \right) \mathrm{e}^{-qt}}{q\left(
qM\left( q\mathrm{e}^{\mathrm{i}\pi }\right) \sinh \left( \kappa qM\left( q%
\mathrm{e}^{\mathrm{i}\pi }\right) \right) +\kappa \cosh \left( \kappa
qM\left( q\mathrm{e}^{\mathrm{i}\pi }\right) \right) \right) }\mathrm{d}q,
\notag \\
&&  \label{Q-plus} \\
\lim_{\substack{ R\rightarrow \infty  \\ r\rightarrow 0}}\int\nolimits_{%
\Gamma _{5}}\tilde{\sigma}_{H}\left( x,s\right) \mathrm{e}^{st}\mathrm{d}s
&=&\int\nolimits_{0}^{\infty }\frac{\kappa \cosh \left( \kappa xqM\left( q%
\mathrm{e}^{-\mathrm{i}\pi }\right) \right) \mathrm{e}^{-qt}}{q\left(
qM\left( q\mathrm{e}^{-\mathrm{i}\pi }\right) \sinh \left( \kappa qM\left( q%
\mathrm{e}^{-\mathrm{i}\pi }\right) \right) +\kappa \cosh \left( \kappa
qM\left( q\mathrm{e}^{-\mathrm{i}\pi }\right) \right) \right) }\mathrm{d}q,
\notag \\
&&  \label{Q-minus} \\
\lim_{R\rightarrow \infty }\int\nolimits_{\gamma _{0}}\tilde{\sigma}%
_{H}\left( x,s\right) \mathrm{e}^{st}\mathrm{d}s &=&2\pi \mathrm{i}\sigma
_{H}\left( x,t\right) .  \label{Q-pi}
\end{eqnarray}%
By the same arguments as in the proof of Theorem 1 we have that (\ref{Q-pi})
is valid if the inversion of the Laplace transform exists. This is true
since all the singularities of $\tilde{\sigma}_{H}$ are left from the line $%
\gamma _{0}$ and appropriate estimates on $\tilde{\sigma}_{H}$ are
satisfied. Adding (\ref{Q-epsilon}), (\ref{Q-plus}), (\ref{Q-minus}) and (%
\ref{Q-pi}) we obtain the left hand side of (\ref{KF-Q}) and finally $\sigma
_{H}$ in the form given by (\ref{Q1}).

Function $\sigma _{H}$ is a sum of three addends: $H$ and%
\begin{eqnarray*}
&&\frac{\kappa }{\pi }\dint\nolimits_{0}^{\infty }\func{Im}\left( \frac{%
\cosh \left( \kappa xqM\left( q\mathrm{e}^{\mathrm{i}\pi }\right) \right) }{%
qM\left( q\mathrm{e}^{\mathrm{i}\pi }\right) \sinh \left( \kappa qM\left( q%
\mathrm{e}^{\mathrm{i}\pi }\right) \right) +\kappa \cosh \left( \kappa
qM\left( q\mathrm{e}^{\mathrm{i}\pi }\right) \right) }\right) \frac{\mathrm{e%
}^{-qt}}{q}\mathrm{d}q, \\
&&2\sum_{n=1}^{\infty }\func{Re}\left( \func{Res}\left( \tilde{\sigma}%
_{H}\left( x,s\right) \mathrm{e}^{st},s_{n}\right) \right) .
\end{eqnarray*}%
As in the case of Theorem \ref{thmP}, the explicit form of solution implies
that $\sigma _{H}$ is continuous on $\left[ 0,1\right] \times \left[
0,\infty \right) .$
\end{proof}

\section{The case of elastic rod\label{ER}}

We treat the case of elastic rod separately. Then, for $s\in V,$ $M\left(
s\right) =1$ ($r\left( s\right) =1$ and $h\left( s\right) \equiv 0$) and
clearly, all the conditions $\left( \mathrm{A1}\right) $ - $\left( \mathrm{A4%
}\right) $ hold. By (\ref{P-tilda}) and (\ref{Q-tilda}) we have
\begin{eqnarray}
\tilde{P}_{el}\left( x,s\right) &=&\frac{1}{s}\frac{\sinh \left( \kappa
sx\right) }{s\sinh \left( \kappa s\right) +\kappa \cosh \left( \kappa
s\right) },\;\;x\in \left[ 0,1\right] ,\;s\in V,  \label{P-t-el} \\
\tilde{Q}_{el}\left( x,s\right) &=&\frac{\kappa \cosh \left( \kappa
sx\right) }{s\sinh \left( \kappa s\right) +\kappa \cosh \left( \kappa
s\right) },\;\;x\in \left[ 0,1\right] ,\;s\in V.  \label{Q-t-el}
\end{eqnarray}

We apply the results of the previous section. Propositions \ref{pr-conj} and %
\ref{pr-conj-besk} imply that the zeros of
\begin{equation*}
f_{el}\left( s\right) :=s\sinh \left( \kappa s\right) +\kappa \cosh \left(
\kappa s\right) =0,\;\;s\in V,
\end{equation*}%
are of the form
\begin{equation*}
s_{n}=\mathrm{i}w_{n},\;\;\tan \left( \kappa w_{n}\right) =\frac{\kappa }{%
w_{n}},\;\;w_{n}\approx \pm \frac{n\pi }{\kappa },\;\;\text{as}%
\;\;n\rightarrow \infty .
\end{equation*}%
Each of these zeros is of multiplicity one for $n>n_{0}$. Moreover, all the
zeros $s_{n},$ $n>n_{0},$ of $f_{el}$ lie on the imaginary axis, so that we
do not have the branch point at $s=0.$ This implies that the integrals over $%
\Gamma _{3}$\ and $\Gamma _{5}$ (see Figure \ref{fig}) are equal to zero. So
we have the following modifications.

\begin{theorem}
\label{thmP-el}Let $F\in \mathcal{S}_{+}^{\prime }.$ Then the unique
solution $u$ to (\ref{sys-1}) - (\ref{BC}) is given by%
\begin{equation*}
u\left( x,t\right) =F\left( t\right) \ast P_{el}\left( x,t\right) ,\;\;x\in
\left[ 0,1\right] ,\;t>0,
\end{equation*}%
where%
\begin{equation}
P_{el}\left( x,t\right) =2\sum_{n=1}^{\infty }\frac{\sin \left( \kappa
w_{n}x\right) \sin \left( w_{n}t\right) }{w_{n}\left( \left( 1+\kappa
^{2}\right) \sin \left( \kappa w_{n}\right) +\kappa w_{n}\cos \left( \kappa
w_{n}\right) \right) },\;\;x\in \left[ 0,1\right] ,\;t>0.  \label{P-el}
\end{equation}

In particular, $u\in C\left( \left[ 0,1\right] ,\mathcal{S}_{+}^{\prime
}\right) .$ Moreover, if $F\in L_{loc}^{1}\left( \left[ 0,\infty \right)
\right) ,$ then $u\in C\left( \left[ 0,1\right] \times \left[ 0,\infty
\right) \right) .$
\end{theorem}

\begin{proof}
The explicit form of $P_{el}$ is obtained from (\ref{P-t-el}) by the use of
the Cauchy residues theorem ($x\in \left[ 0,1\right] ,$ $t>0$)%
\begin{equation}
\frac{1}{2\pi \mathrm{i}}\oint\nolimits_{\Gamma _{el}}\tilde{P}_{el}\left(
x,s\right) \mathrm{e}^{st}\mathrm{d}s=\sum_{n=1}^{\infty }\left( \func{Res}%
\left( \tilde{P}_{el}\left( x,s\right) \mathrm{e}^{st},s_{n}\right) +\func{%
Res}\left( \tilde{P}_{el}\left( x,s\right) \mathrm{e}^{st},\bar{s}%
_{n}\right) \right) ,  \label{crt}
\end{equation}%
where the integration contour $\Gamma _{el}=\Gamma _{1}\cup \Gamma _{2}\cup
\Gamma _{3}\cup \gamma _{0}$ is presented in Figure \ref{figg}.
\begin{figure}[h]
\centering
\includegraphics[scale=0.45]{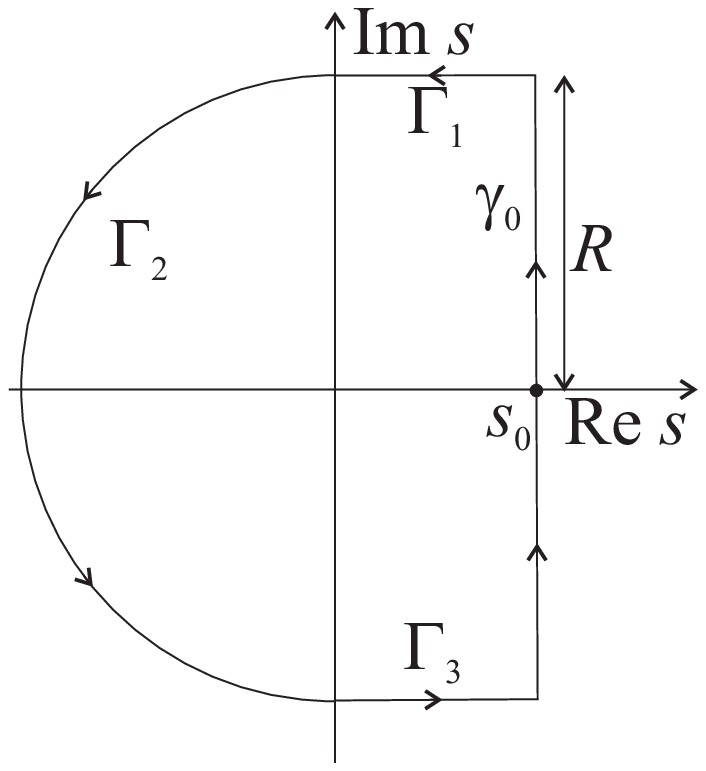}
\caption{Integration contour $\Gamma _{el}$}
\label{figg}
\end{figure}

First we show that the series of residues in (\ref{crt}) is real-valued and
convergent. Since the poles $s_{n}=\mathrm{i}w_{n}$ (and $\bar{s}_{n}=-%
\mathrm{i}w_{n}$) of $\tilde{P}_{el}\left( x,s\right) \mathrm{e}^{st}$ are
simple for $n>n_0$, the residues in (\ref{crt}) are calculated by ($x\in %
\left[ 0,1\right] ,$ $t>0$)%
\begin{eqnarray*}
\func{Res}\left( \tilde{P}_{el}\left( x,s\right) \mathrm{e}%
^{st},s_{n}\right) &=&\left[ \frac{1}{s}\frac{\sinh \left( \kappa sx\right)
\mathrm{e}^{st}}{\frac{\mathrm{d}}{\mathrm{d}s}\left( s\sinh \left( \kappa
s\right) +\kappa \cosh \left( \kappa s\right) \right) }\right] _{s=\mathrm{i}%
w_{n}} \\
&=&\frac{\sin \left( \kappa w_{n}x\right) \left( \sin \left( w_{n}t\right) -%
\mathrm{i}\cos \left( w_{n}t\right) \right) }{w_{n}\left( \left( 1+\kappa
^{2}\right) \sin \left( \kappa w_{n}\right) +\kappa w_{n}\cos \left( \kappa
w_{n}\right) \right) }.
\end{eqnarray*}%
Since $w_{n}\approx \pm \frac{n\pi }{\kappa },$ as $n\rightarrow \infty ,$
the previous expression becomes
\begin{equation*}
\left\vert \func{Res}\left( \tilde{P}_{el}\left( x,s\right) \mathrm{e}%
^{st},s_{n}\right) \right\vert \leq \frac{k}{n^{2}},\;\;x\in \left[ 0,1%
\right] ,\;t>0, n>n_0.
\end{equation*}%
We conclude that the series of residues is convergent.

Now, we calculate the integral over $\Gamma $ in (\ref{crt}). First, we
consider the integral along contour $\Gamma _{1}=\left\{ s=p+\mathrm{i}R\mid
p\in \left[ 0,s_{0}\right] ,\;R>0\right\} ,$ where $R$ is defined as
\begin{equation}
R=\frac{n\pi }{\kappa }+\frac{1}{2}\frac{\pi }{\kappa },\;\;n>n_{0}.
\label{er}
\end{equation}%
Let $x\in \left[ 0,1\right] ,$ $t>0$. By (\ref{P-t-el}) and Corollary \ref%
{zbognjegasveovo}, we have
\begin{equation}
\left\vert \tilde{P}_{el}\left( x,s\right) \right\vert \leq \frac{C}{%
\left\vert s\right\vert ^{2}},\;\;\left\vert s\right\vert \rightarrow \infty
.  \label{P-s-besk-el}
\end{equation}%
Using (\ref{P-s-besk-el}), we calculate the integral over $\Gamma _{1}$ as
\begin{eqnarray*}
\lim_{R\rightarrow \infty }\left\vert \int_{\Gamma _{1}}\tilde{P}_{el}\left(
x,s\right) \mathrm{e}^{st}\mathrm{d}s\right\vert &\leq &\lim_{R\rightarrow
\infty }\int_{0}^{s_{0}}\left\vert \tilde{P}_{el}\left( x,p+\mathrm{i}%
R\right) \right\vert \left\vert \mathrm{e}^{\left( p+\mathrm{i}R\right)
t}\right\vert \mathrm{d}p \\
&\leq &C\lim_{R\rightarrow \infty }\int_{0}^{s_{0}}\frac{1}{R^{2}}\mathrm{e}%
^{pt}\mathrm{d}p=0.
\end{eqnarray*}%
Similar arguments are valid for the integral along $\Gamma _{3}.$ Thus, we
have%
\begin{equation*}
\lim\limits_{R\rightarrow \infty }\left\vert \int\nolimits_{\Gamma _{3}}%
\tilde{P}_{el}\left( x,s\right) \mathrm{e}^{st}\mathrm{d}s\right\vert
=0,\;\;x\in \left[ 0,1\right] ,\;t>0.
\end{equation*}%
Next, we consider the integral along contour $\Gamma _{2}=\left\{ s=R\,%
\mathrm{e}^{\mathrm{i}\phi }\mid \phi \in \left[ -\frac{\pi }{2},\frac{\pi }{%
2}\right] \right\} .$ By using (\ref{P-s-besk-el}), the integral over the
contour $\Gamma _{2}$ becomes%
\begin{eqnarray*}
\lim_{R\rightarrow \infty }\left\vert \int\nolimits_{\Gamma _{2}}\tilde{P}%
_{el}\left( x,s\right) \mathrm{e}^{st}\mathrm{d}s\right\vert &\leq
&\lim_{R\rightarrow \infty }\int\nolimits_{-\frac{\pi }{2}}^{\frac{\pi }{2}%
}\left\vert \tilde{P}_{el}\left( x,R\mathrm{e}^{\mathrm{i}\phi }\right)
\right\vert \left\vert \mathrm{e}^{Rt\mathrm{e}^{\mathrm{i}\phi
}}\right\vert \left\vert \mathrm{i}R\mathrm{e}^{\mathrm{i}\phi }\right\vert
\mathrm{d}\phi \\
&\leq &C\lim_{R\rightarrow \infty }\int\nolimits_{-\frac{\pi }{2}}^{\frac{%
\pi }{2}}\frac{1}{R}\mathrm{e}^{Rt\cos \phi }\mathrm{d}\phi =0,\;\;x\in %
\left[ 0,1\right] ,\;t>0,
\end{eqnarray*}%
since $\cos \phi \leq 0$ for $\phi \in \left[ -\frac{\pi }{2},\frac{\pi }{2}%
\right] .$ Integrating along the Bromwich contour, we have%
\begin{equation*}
\lim_{R\rightarrow \infty }\int_{\gamma _{0}}\tilde{P}_{el}\left( x,s\right)
\mathrm{e}^{st}\mathrm{d}s=2\pi \mathrm{i}P_{el}\left( x,t\right) ,\;\;x\in %
\left[ 0,1\right] ,\;t>0.
\end{equation*}

Therefore, (\ref{crt}) yields $P_{el}$ in the form (\ref{P-el}). The last
assertion of the theorem follows from the proof of Theorem \ref{thmP}.
\end{proof}

\begin{theorem}
\label{thmQ-el}Let $F=H.$ Then the unique solution $\sigma _{H}^{\left(
el\right) }$ to (\ref{sys-1}) - (\ref{BC}), is given by%
\begin{equation}
\sigma _{H}^{\left( el\right) }\left( x,t\right) =-2\kappa
\sum_{n=1}^{\infty }\frac{\cos \left( \kappa w_{n}x\right) \cos \left(
w_{n}t\right) }{w_{n}\left( \left( 1+\kappa ^{2}\right) \sin \left( \kappa
w_{n}\right) +\kappa w_{n}\cos \left( \kappa w_{n}\right) \right) },\;\;x\in %
\left[ 0,1\right] ,\;t>0.  \label{sigma-el}
\end{equation}%
In particular, $\sigma _{H}^{\left( el\right) }$ is continuous on $\left[ 0,1%
\right] \times \left[ 0,\infty \right) .$
\end{theorem}

\begin{proof}
The explicit forms of $\sigma _{H}^{\left( el\right) }$ is obtained from (%
\ref{Q-t-el}) and (\ref{u,sigma-tilda})$_{2},$ with $\tilde{F}=\frac{1}{s},$
i.e.
\begin{equation*}
\tilde{\sigma}_{H}^{\left( el\right) }\left( x,s\right) =\frac{1}{s}\frac{%
\kappa \cosh \left( \kappa sx\right) }{s\sinh \left( \kappa s\right) +\kappa
\cosh \left( \kappa s\right) },\;\;x\in \left[ 0,1\right] ,\;s\in V,
\end{equation*}%
by the use of the Cauchy residues theorem ($x\in \left[ 0,1\right] ,$ $t>0$)%
\begin{equation}
\frac{1}{2\pi \mathrm{i}}\oint\nolimits_{\Gamma _{el}}\tilde{\sigma}%
_{H}^{\left( el\right) }\left( x,s\right) \mathrm{e}^{st}\mathrm{d}%
s=\sum_{n=1}^{\infty }\left( \func{Res}\left( \tilde{\sigma}_{H}^{\left(
el\right) }\left( x,s\right) \mathrm{e}^{st},s_{n}\right) +\func{Res}\left(
\tilde{\sigma}_{H}^{\left( el\right) }\left( x,s\right) \mathrm{e}^{st},\bar{%
s}_{n}\right) \right) ,  \label{crt-s}
\end{equation}%
where the integration contour $\Gamma _{el}$ is the contour from Figure \ref%
{figg}.

First we show that the series of residues in (\ref{crt-s}) is convergent and
real-valued. Since the poles $s_{n}=\mathrm{i}w_{n}$ (and $\bar{s}_{n}=-%
\mathrm{i}w_{n}$) of $\tilde{\sigma}_{H}^{\left( el\right) }\left(
x,s\right) \mathrm{e}^{st}$ are simple for $n>n_{0}$, the residues in (\ref%
{crt-s}) are calculated by%
\begin{eqnarray*}
\func{Res}\left( \tilde{\sigma}_{H}^{\left( el\right) }\left( x,s\right)
\mathrm{e}^{st},s_{n}\right) &=&\left[ \frac{1}{s}\frac{\kappa \cosh \left(
\kappa sx\right) \mathrm{e}^{st}}{\frac{\mathrm{d}}{\mathrm{d}s}\left(
s\sinh \left( \kappa s\right) +\kappa \cosh \left( \kappa s\right) \right) }%
\right] _{s=\mathrm{i}w_{n}} \\
&=&-\frac{\kappa \cos \left( \kappa w_{n}x\right) \left( \cos \left(
w_{n}t\right) +\mathrm{i}\sin \left( w_{n}t\right) \right) }{w_{n}\left(
\left( 1+\kappa ^{2}\right) \sin \left( \kappa w_{n}\right) +\kappa
w_{n}\cos \left( \kappa w_{n}\right) \right) },\;n>n_{0}.
\end{eqnarray*}%
As $n\rightarrow \infty $ $w_{n}\approx \pm \frac{n\pi }{\kappa },$ so that
the previous expression becomes%
\begin{equation*}
\left\vert \func{Res}\left( \tilde{\sigma}_{H}^{\left( el\right) }\left(
x,s\right) \mathrm{e}^{st},s_{n}\right) \right\vert \leq \frac{k}{n^{2}}%
,\;\;x\in \left[ 0,1\right] ,\;t>0,\;n>n_{0}.
\end{equation*}%
We conclude that the series of residues is convergent.

Let us calculate the integral over $\Gamma _{el}$ in (\ref{crt-s}). Consider
the integral along contour $\Gamma _{1}=\left\{ s=p+\mathrm{i}R\mid p\in %
\left[ 0,s_{0}\right] ,\;R>0\right\} ,$ where $R$ is defined by (\ref{er}).
We use estimates%
\begin{equation}
\frac{\left\vert \cosh \left( \kappa xs\right) \right\vert }{\left\vert
f_{el}(s)\right\vert }\leq \frac{C}{\left\vert s\right\vert },\;\;\text{or}%
\;\;\frac{\left\vert \cosh \left( \kappa xs\right) \right\vert }{\left\vert
f_{el}(s)\right\vert }\leq C,  \label{estim-el}
\end{equation}%
from Corollary \ref{IIzbognjegasveovo} in (\ref{sigma-el}). With the first
estimate in (\ref{estim-el}) we calculate the integral over $\Gamma _{1}$ as
\begin{eqnarray*}
\lim_{R\rightarrow \infty }\left\vert \int\nolimits_{\Gamma _{1}}\tilde{%
\sigma}_{H}^{\left( el\right) }\left( x,s\right) \mathrm{e}^{st}\mathrm{d}%
s\right\vert &\leq &\lim_{R\rightarrow \infty }\int_{0}^{s_{0}}\left\vert
\tilde{\sigma}_{H}^{\left( el\right) }\left( x,p+\mathrm{i}R\right)
\right\vert \left\vert \mathrm{e}^{\left( p+\mathrm{i}R\right) t}\right\vert
\mathrm{d}p \\
&\leq &C\kappa \lim_{R\rightarrow \infty }\int_{0}^{s_{0}}\frac{1}{R^{2}}%
\mathrm{e}^{pt}\mathrm{d}p=0,\;\;x\in \left[ 0,1\right] ,\;t>0,
\end{eqnarray*}%
while with the second estimate in (\ref{estim-el}), we have%
\begin{eqnarray*}
\lim_{R\rightarrow \infty }\left\vert \int\nolimits_{\Gamma _{1}}\tilde{%
\sigma}_{H}^{\left( el\right) }\left( x,s\right) \mathrm{e}^{st}\mathrm{d}%
s\right\vert &\leq &\lim_{R\rightarrow \infty }\int_{0}^{s_{0}}\left\vert
\tilde{\sigma}_{H}^{\left( el\right) }\left( x,p+\mathrm{i}R\right)
\right\vert \left\vert \mathrm{e}^{\left( p+\mathrm{i}R\right) t}\right\vert
\mathrm{d}p \\
&\leq &C\kappa \lim_{R\rightarrow \infty }\int_{0}^{s_{0}}\frac{1}{R}\mathrm{%
e}^{pt}\mathrm{d}p=0,\;\;x\in \left[ 0,1\right] ,\;t>0.
\end{eqnarray*}%
Similar arguments are valid for the integral along $\Gamma _{3}.$ Thus, we
have%
\begin{equation*}
\lim\limits_{R\rightarrow \infty }\left\vert \int\nolimits_{\Gamma 3}\tilde{%
\sigma}_{H}^{\left( el\right) }\left( x,s\right) \mathrm{e}^{st}\mathrm{d}%
s\right\vert =0,\;\;x\in \left[ 0,1\right] ,\;t>0.
\end{equation*}%
Next, we consider the integral along contour $\Gamma _{2}=\left\{ s=R\,%
\mathrm{e}^{\mathrm{i}\phi }\mid \phi \in \left[ -\frac{\pi }{2},\frac{\pi }{%
2}\right] \right\} .$ With the first estimate in (\ref{estim-el}) we have
that the integral over $\Gamma _{2}$ becomes
\begin{eqnarray*}
\lim_{R\rightarrow \infty }\left\vert \int\nolimits_{\Gamma _{2}}\tilde{%
\sigma}_{H}^{\left( el\right) }\left( x,s\right) \mathrm{e}^{st}\mathrm{d}%
s\right\vert &\leq &\lim_{R\rightarrow \infty }\int\nolimits_{-\frac{\pi }{2}%
}^{\frac{\pi }{2}}\left\vert \tilde{\sigma}_{H}^{\left( el\right) }\left( x,R%
\mathrm{e}^{\mathrm{i}\phi }\right) \right\vert \left\vert \mathrm{e}^{Rt%
\mathrm{e}^{\mathrm{i}\phi }}\right\vert \left\vert \mathrm{i}R\mathrm{e}^{%
\mathrm{i}\phi }\right\vert \mathrm{d}\phi \\
&\leq &C\kappa \lim_{R\rightarrow \infty }\int\nolimits_{-\frac{\pi }{2}}^{%
\frac{\pi }{2}}\frac{1}{R}\mathrm{e}^{Rt\cos \phi }\mathrm{d}\phi
=0,\;\;x\in \left[ 0,1\right] ,\;t>0,
\end{eqnarray*}%
since $\cos \phi \leq 0$ for $\phi \in \left[ -\frac{\pi }{2},\frac{\pi }{2}%
\right] .$ The integral over $\Gamma _{2},$ in the case of the second
estimate in (\ref{estim-el}) becomes%
\begin{eqnarray*}
\lim_{R\rightarrow \infty }\left\vert \int\nolimits_{\Gamma _{2}}\tilde{%
\sigma}_{H}\left( x,s\right) \mathrm{e}^{st}\mathrm{d}s\right\vert &\leq
&\lim_{R\rightarrow \infty }\int\nolimits_{-\frac{\pi }{2}}^{\frac{\pi }{2}%
}\left\vert \tilde{\sigma}_{H}\left( x,R\mathrm{e}^{\mathrm{i}\phi }\right)
\right\vert \left\vert \mathrm{e}^{Rt\mathrm{e}^{\mathrm{i}\phi
}}\right\vert \left\vert \mathrm{i}R\mathrm{e}^{\mathrm{i}\phi }\right\vert
\mathrm{d}\phi \\
&\leq &C\kappa \lim_{R\rightarrow \infty }\int\nolimits_{-\frac{\pi }{2}}^{%
\frac{\pi }{2}}\mathrm{e}^{Rt\cos \phi }\mathrm{d}\phi =0,\;\;x\in \left[ 0,1%
\right] ,\;t>0,
\end{eqnarray*}%
since $\cos \phi \leq 0$ for $\phi \in \left[ -\frac{\pi }{2},\frac{\pi }{2}%
\right] .$ Integrating along the Bromwich contour, we have%
\begin{equation*}
\lim_{R\rightarrow \infty }\int_{\gamma _{0}}\tilde{\sigma}_{H}^{\left(
el\right) }\left( x,s\right) \mathrm{e}^{st}\mathrm{d}s=2\pi \mathrm{i}%
\sigma _{H}^{\left( el\right) }\left( x,t\right) ,\;\;x\in \left[ 0,1\right]
,\;t>0.
\end{equation*}

Therefore, (\ref{crt-s}) yields $\sigma _{H}^{\left( el\right) }$ in the
form (\ref{sigma-el}). The last assertion of the theorem follows from the
proof of Theorem \ref{thmQ}.
\end{proof}

\begin{remark}
Numerical analysis for specific choices of constitutive equations, as well
as the interpretation of the results is given in \cite{APZ-6}.
\end{remark}

\begin{acknowledgement}
This research is supported by the Serbian Ministry of Education and Science
projects $174005$ (TMA and DZ) and $174024$ (SP), as well as by the
Secretariat for Science of Vojvodina project $114-451-2167$ (DZ).
\end{acknowledgement}

\end{document}